\documentclass[a4paper,twocolumn,11pt,accepted=2019-07-07]{quantumarticle}
\pdfoutput=1

\usepackage{amsmath,amsthm,amsfonts,amssymb,amscd} 
\usepackage[utf8]{inputenc}
\usepackage{indentfirst}
\usepackage{graphicx}
\usepackage{color}
\usepackage[T1]{fontenc}
\usepackage[american,ngerman,greek,brazilian,british]{babel}
\usepackage[dvipsnames]{xcolor}

\definecolor{quantumviolet}{HTML}{53257F}

\usepackage[colorlinks=true,citecolor=Emerald,linkcolor=magenta,urlcolor=quantumviolet,hyperindex]{hyperref}
\usepackage{cleveref}
\usepackage{mathtools}
\usepackage{bm}
\usepackage{ulem}
\usepackage{enumerate}
\usepackage{sidecap}
\usepackage{microtype}
\usepackage{enumitem}
\usepackage{bbold}
\usepackage{listings}
\usepackage{algorithm}
\usepackage{algpseudocode}
\usepackage{multirow}

\usepackage{soul}

\newcommand{\ket}[1]{\vert#1\rangle}

\newcommand{\ketbra}[2]{\vert #1 \rangle \langle #2 \vert}

\newcommand{\1}{\mathbb{1}}
\newcommand{\id}{\mathbb{1}}
\newcommand{\set}[1]{\mathcal{#1}}

\newcommand{\proj}[1]{|#1\rangle\langle#1|}

\newcommand{\de}[1]{\left(#1\right)}
\newcommand{\De}[1]{\left[#1\right]}

\newcommand{\tr}{\text{\normalfont Tr}}

\newcommand{\beq}{\begin{equation}}
\newcommand{\eeq}{\end{equation}}

\newcommand{\causalc}[1]{\textcolor{Emerald}{#1}}
\newcommand{\noncausalc}[1]{\textcolor{Mulberry}{#1}}

\newtheorem{theorem}{Theorem}
\newtheorem*{theorem*}{Theorem}
\newtheorem{definition}{Definition}
\newtheorem{proposition}{Proposition}
\newtheorem{lemma}{Lemma}

\newcommand{\iqoqi}{Institute for Quantum Optics and Quantum Information (IQOQI), Austrian Academy of Sciences, Boltzmanngasse 3, A-1090 Vienna, Austria}
\newcommand{\koeln}{Institute for Theoretical Physics, University of Cologne, Z\"ulpicher Strasse 77, 50937 Cologne, Germany}
\newcommand{\univie}{Vienna Center for Quantum Science and Technology (VCQ), Faculty of Physics, \\ University of Vienna, Boltzmanngasse 5, A-1090 Vienna, Austria}
\newcommand{\todai}{Department of Physics, Graduate School of Science, The University of Tokyo, Hongo 7-3-1, Bunkyo-ku, Tokyo 113-0033, Japan}

\begin{document}

\title{Semi-device-independent certification of indefinite causal order} 

\author{Jessica Bavaresco}
\email{jessica.bavaresco@oeaw.ac.at}
\affiliation{\iqoqi}
\orcid{0000-0002-4823-0353}

\author{Mateus Ara\'ujo}
\affiliation{\koeln}

\author{\v{C}aslav Brukner}
\affiliation{\iqoqi}
\affiliation{\univie}
\orcid{0000-0002-6549-5863}

\author{Marco T\'ulio Quintino}
\email{quintino@eve.phys.s.u-tokyo.ac.jp}
\affiliation{\todai}


\begin{abstract}
When transforming pairs of independent quantum operations according to the fundamental rules of quantum theory, an intriguing phenomenon emerges: some such higher-order operations may act on the input operations in an indefinite causal order. Recently, the formalism of process matrices has been developed to investigate these noncausal properties of higher-order operations. 
This formalism predicts, in principle, statistics that ensure indefinite causal order even in a device-independent scenario, where the involved operations are not characterised. Nevertheless, all physical implementations of process matrices proposed so far require full characterisation of the involved operations in order to certify such phenomena. 
Here we consider a semi-device-independent scenario, which does not require all operations to be characterised. We introduce a framework for certifying noncausal properties of process matrices in this intermediate regime and use it to analyse the quantum switch, a well-known higher-order operation, to show that, although it can only lead to causal statistics in a device-independent scenario, it can exhibit noncausal properties in semi-device-independent scenarios. This proves that the quantum switch generates stronger noncausal correlations than it was previously known.
\end{abstract}

\maketitle


A common quantum information task consists in certifying that some uncharacterised source is preparing a system with some features. By making the assumption that the measurement devices are completely characterised, that is, that they are known exactly, it is possible to infer properties of the system. In this \textit{device-dependent} scenario, fidelity of a quantum state with respect to a target state can be estimated, entanglement witnesses can be evaluated \cite{horodecki96}, and even complete characterisation of the source via state tomography is possible \cite{dariano02}.

Remarkably, it is possible to certify properties of systems even without fully characterizing the measurement devices \cite{bell64,brunner14}. In such a \textit{device-independent} scenario it is only assumed that the measurements are done by separated parties and compose under a tensor product, which is justified by implementing them with a space-like separation. Under these circumstances, Bell scenarios can be used to certify properties like entanglement of quantum states \cite{brunner14}, incompatibility of quantum measurements \cite{wolf09}, or to perform device-independent state estimation via self-testing \cite{mayers04,bardyn09}. 

Since the assumptions are weaker, demonstrations of device-independent certification are usually experimentally challenging. For instance, although experimental device-independent certification of entanglement has been reported \cite{aspect82,hensen15,giustina15,shalm15}, its experimental difficulty has so far prevented its use in practical applications such as device-independent quantum key distribution \cite{acin07} and randomness certification \cite{acin11}.

An interesting middle ground is the \textit{semi-device-independent} scenario, where assumptions are made about some parties but not others. Semi-device-independent schemes have been developed and extensively studied for the certification of entanglement \cite{wiseman07} and measurement incompatibility \cite{quintino14,uola14}, known as EPR-steering, and applied to quantum key distribution protocols where some but not all parties can be trusted \cite{branciard12}.

A close analogy can be developed with regard to the certification of indefinite causal order, as encoded in a \textit{process matrix} \cite{oreshkov12}. A process matrix is a higher-order operation \cite{chiribella08,chiribella09,araujo17} -- \textit{i.e.} a transformation of quantum operations -- that acts on independent sets of operations. Fundamental laws of quantum theory predict the existence of process matrices that act on these operations in a such a way that a well-defined causal order cannot be established among them. Process matrices with indefinite causal order were proven to be a powerful resource, outperforming causally ordered ones in tasks such as quantum channel discrimination \cite{chiribella12}, communication complexity \cite{feix15,guerin16}, quantum computation \cite{araujo14}, and inverting unknown unitary operations \cite{quintino18}.

To certify that a process matrix in fact does not act in a causally ordered way, there are two standard methods available in the literature. The first is to evaluate a causal witness \cite{araujo15,branciard16-2}. Analogous to the evaluation of an entanglement witness, this method relies on detailed knowledge of the quantum operations being implemented, and as such it allows for a device-dependent certification. All experimental certifications of indefinite causal order to date either measure a causal witness \cite{rubino17,goswami18-1} or rely on similar device-dependent assumptions \cite{procopio15,rubino17-2,goswami18-2}. The second method is the violation of a causal inequality, phenomenon which is also predicted by quantum mechanics \cite{oreshkov12,branciard16}. Analogous to the violation of a Bell inequality, this method does not rely on detailed knowledge of the quantum operations implemented by the parties, but rather only that they compose under a tensor product. As such, it allows for a device-independent certification. When a causal inequality is violated, it is verified that least one of the principles used to derived it is not respected. Moreover, this claim holds true independently of the physical theory that supports the experiment that led to the violation. Hence, although we focus on quantum theory, indefinite causal order could in principle be certified even without relying on the laws of quantum mechanics. Although it would be highly desirable to perform such device-independent certification of indefinite causal order, no physical implementation of process matrices that would violate a causal inequality is currently known. 

In this work, we introduce a semi-device-independent framework for certifying noncausal properties of process matrices that allows for an experimental certification of indefinite causal order that relies on fewer assumptions than previous ones. In our semi-device-independent scenario, the operations of some parties are fully characterized while no assumptions are made about the others. 

{We begin by considering the bipartite case, for which we construct a general framework for certifying indefinite causal order in a semi-device-independent scenario, while contextualizing previously developed device-dependent and -independent ones. We then extend our framework to a tripartite case in which the third party is always in the future of the other two, and provide an extensive machinery that may be generalized to other multipartite scenarios. We apply our methods to} the notorious quantum switch \cite{chiribella12,chiribella13}, a process matrix that, despite having an indefinite causal order, only leads to causal correlations in a device-independent scenario. We show that the noncausal properties of the quantum switch can be certified in a semi-device-independent way, proving that it generates stronger noncausal correlations than it was previously known.


\section{Preliminaries}\label{sec:preliminaries}

In our certification scheme, we will deal with statistical data in the form of behaviours.

A \textit{general bipartite behaviour} $\{p(ab|xy)\}$ is a set of joint probability distributions, that is, a set in which each element $p(ab|xy)$ is a real non-negative number such that $\sum_{a,b} p(ab|xy) = 1$ for all $x,y$, where $a\in\{1,\ldots,O_A\}$ and $b\in\{1\ldots,O_B\}$ are labels for outcomes and $x\in\{1,\ldots,I_A\}$ and $y\in\{1,\ldots,I_B\}$ are labels for inputs, for parties Alice and Bob, respectively.

The most general operation allowed by quantum theory is modelled by a quantum instrument, which is a set of completely-positive (CP) maps that sum to a completely-positive trace-preserving (CPTP) map \cite{kraus83}. 

The Choi-Jamio\l kowski (CJ) isomorphism \cite{depillis67,jamiolkowski72,choi75} allows us to represent every linear map\footnote{In this paper we only consider finite dimensional complex linear spaces. That is, all linear spaces are isomorphic to $\mathbb{C}^d$ for some natural number $d$.} $\set{M}: \set{L}(\set{H}^I)\to\set{L}(\set{H}^O)$ by a linear operator $M\in\set{L}(\set{H}^I\otimes\set{H}^O)$ acting on the joint input and output Hilbert spaces. In this representation, a \textit{set of instruments} is a set of operators $\{I_{a|x}\}$, $I_{a|x}\in\set{L}(\set{H}^{I}\otimes\set{H}^{O})$, that satisfies
\begin{align}
I_{a|x}&\geq0, \ \ \ \forall \ a,x \\
\tr_{O}\sum_a I_{a|x} &= \1^{I}, \ \ \ \forall \ x,
\end{align}
where $x\in\{1,\ldots,I\}$ labels the instrument in the set and $a\in\{1,\ldots,O\}$ its outcomes, and $\1^{I}$ is the identity operator on $\set{H}^{I}$.

Now let us consider the most general set of behaviours which respects quantum theory. For that we follow the steps of ref.~\cite{oreshkov12} to analyse behaviours that can be extracted by pairs of independent quantum instruments. Let $\{A_{a|x}\}$, $A_{a|x}\in\set{L}(\set{H}^{A_I}\otimes\set{H}^{A_O})$ and $\{B_{b|y}\}$, $B_{b|y}\in\set{L}(\set{H}^{B_I}\otimes\set{H}^{B_O})$ be the Choi operators of Alice's and Bob's local instruments. We then seek a function which assigns probabilities to a pair of instrument elements $A_{a|x}$ and $B_{b|y}$. In order to preserve the structure of quantum mechanics, we assume that this function is linear in both arguments. It follows from the Riesz representation lemma \cite{riesz1907} that this general linear function necessarily has the form of $p(ab|xy) = \tr\left[(A^{A_IA_O}_{a|x}\otimes B^{B_IB_O}_{b|y})\,W\right]$ for some linear operator $W\in\set{L}(\set{H}^{A_I}\otimes\set{H}^{A_O}\otimes\set{H}^{B_I}\otimes\set{H}^{B_O})$. In order to be consistent with extended quantum scenarios, we also consider the case where Alice and Bob may share a (potentially entangled) auxiliary quantum state $\rho\in \set{L}(\set{H}^{A_{I'}}\otimes \set{H}^{B_{I'}})$, and have instruments $\{A'_{a|x}\}$, $A'_{a|x} \in \set{L}(\set{H}^{A_{I'}A_IA_O})$, $\{B'_{b|y}\}$, $B'_{b|y} \in \set{L}(\set{H}^{B_{I'}B_IB_O})$ which acts on the space of the operator $W$ and the auxiliary state $\rho$. A \textit{process matrix} is then defined as the most general linear operator $W$ such that 
\footnotesize
\begin{align}\label{eq:X}
\begin{split}
&p(ab|xy) =  \\
&\tr\left[(A'^{A_{I'}A_IA_O}_{a|x}\otimes B'^{B_{I'}B_IB_O}_{b|y}) (W^{A_IA_OB_IB_O}\otimes\rho^{A_{I'}B_{I'}})\right]
\end{split}
\end{align}
\normalsize
represents\footnote{{In} \cref{eq:X}, as well as {in} some other equations, we add superscripts on the operators to indicate the Hilbert spaces in which they act, for sake of clarity.} elements of valid probability distributions for every state $\rho$ and sets of instruments $\{A'_{a|x}\}$ and $\{B'_{b|y}\}$.

It was shown in ref.~\cite{araujo15} that a linear operator $W$ is a process matrix if and only if it respects
\begin{align}
W &\geq 0 \label{eqW1} \\ 
\tr\ W &= d_{A_O}d_{B_O} \label{eqW2} \\ 
_{A_IA_O}W&=_{A_IA_OB_O}W \label{eqW3} \\
_{B_IB_O}W&=_{A_OB_IB_O}W \label{eqW4} \\
W &= _{A_O}W+_{B_O}W-_{A_OB_O}W, \label{eqW5}
\end{align}
where $_XW\coloneqq\tr_XW\otimes\frac{\1^X}{d_X}$ is the trace-and-replace operation and $d_X=\text{dim}(\set{H}^X)$. 

We then define \textit{process behaviours} $\{p^Q(ab|xy)\}$ as behaviours which can be obtained by process matrices according to
\small
\beq\label{eq:processbehaviour}
p^Q(ab|xy) = \tr\left[(A^{A_IA_O}_{a|x}\otimes B^{B_IB_O}_{b|y})\ W^{A_IA_OB_IB_O}\right],
\eeq 
\normalsize
for all $a,b,x,y$.

Now we begin to discuss the causal properties of behaviours and process matrices.

A behaviour is considered causally ordered when it can be established that one party acted before the other because the marginal probability distributions of one party do not depend on the inputs of the other. Formally, we have that a behaviour $\{p^{A\prec B}(ab|xy)\}$ is causally ordered from Alice to Bob if it satisfies
\beq
\sum_b p^{A\prec B}(ab|xy)=\sum_b p^{A\prec B}(ab|xy'),
\eeq
for all $a,x,y,y'$, and equivalently from Bob to Alice.

Behaviours that are within the convex hull of causally ordered behaviours are also considered causal, as they can be interpreted as a classical mixture of causally ordered behaviours.
Hence, a \textit{causal behaviour} $\{p^{\text{causal}}(ab|xy)\}$ is a behaviour that can be expressed as a convex combination of causally ordered behaviours, \textit{i.e.},
\footnotesize
\beq\label{eqcausalP}
p^{\text{causal}}(ab|xy) \coloneqq q p^{A\prec B}(ab|xy) + (1-q) p^{B\prec A}(ab|xy),
\eeq
\normalsize
for all $a,b,x,y$, where $0 \leq q \leq1$ is a real number. Behaviours that do not satisfy \cref{eqcausalP} are called noncausal behaviours.

Following the same reasoning as the one in the definition of a general process matrix, in order to associate causal properties to process matrices we now define causally ordered process matrices as the most general operator that takes pairs of local instruments to causally ordered behaviours, that is,
\beq\label{eqdefcausalW}
p^{A\prec B}(ab|xy) = \tr\left[(A_{a|x}\otimes B_{b|y})\ W^{A\prec B}\right],
\eeq
for all $a,b,x,y$.
This definition is equivalent to the one in refs.~\cite{chiribella09,oreshkov12}, which states that a bipartite process matrix $W^{A\prec B}\in\set{L}(\set{H}^{A_I}\otimes\set{H}^{A_O}\otimes\set{H}^{B_I}\otimes\set{H}^{B_O})$ is causally ordered from Alice to Bob if it satisfies
\beq
W^{A\prec B}=_{B_O}W^{A\prec B}, \label{eqW6}
\eeq
and equivalently from Bob to Alice.

In line with the definition of a causal behaviour, a \textit{causally separable process matrix} 
$W^{\text{sep}}\in\set{L}(\set{H}^{A_I}\otimes\set{H}^{A_O}\otimes\set{H}^{B_I}\otimes\set{H}^{B_O})$ is a process matrix that can be expressed as a convex combination of causally ordered process matrices, \textit{i.e.},
\beq\label{eqcausalW}
W^{\text{sep}} \coloneqq q W^{A\prec B} + (1-q) W^{B\prec A},
\eeq
where $0 \leq q \leq1$ is a real number. Process matrices that do not satisfy \cref{eqcausalW} are called causally nonseparable process matrices.


\section{Certification}\label{sec:certification}

Let us consider the following task: we are given a behaviour that describes the statistics of a quantum experiment. We analyse this behaviour in the process matrix formalism, that is, we assume that there exists a process matrix $W$ and sets of local instruments that give rise to this behaviour according to the rules of quantum theory. Without any information about $W$ -- \textit{i.e.}, without direct assumptions about the process matrix -- the goal is to verify whether it is causally nonseparable. Additionally, information about the instruments which were performed may or may not be given.

The assumptions about the instruments can be split in three: device-dependent, -independent, and semi-device-independent. A device-dependent certification scenario is one in which the operations of all parties are fully characterised, \textit{i.e.}, the whole matrix description of the elements of all applied instruments is known. A device-independent certification scenario is the opposite, no knowledge or assumption is made regarding the operations performed by any parties, not even the dimension of the linear spaces used to describe them. Finally, a semi-device-independent certification scenario is one in which at least one party is device-dependent, which is often called trusted, and at least one is device-independent, often called untrusted.

In the following we formalise our notions of certification for bipartite process matrices. We refer to \cref{app:tripartite} and to \cref{sec:switch} for a discussion of more general scenarios.

\begin{definition}[Device-dependent certification]\label{defDDcert}
Given a process behaviour $\{p^Q(ab|xy)\}$, that arises from known instruments $\{\overline{A}_{a|x}\}$ and $\{\overline{B}_{b|y}\}$ and an unknown bipartite process matrix, one certifies that this process matrix is causally nonseparable in a device-dependent way if, for some $a,b,x,y$,
\beq\label{eqDDcert}
p^Q(ab|\overline{A}_{a|x}\,,\overline{B}_{b|y})\neq\tr\left[(\overline{A}_{a|x}\otimes \overline{B}_{b|y})W^\text{sep}\right],
\eeq
for all causally separable process matrices $W^\text{sep}$.
\end{definition}

\begin{definition}[Device-independent certification]\label{defDIcert}
Given a process behaviour $\{p^Q(ab|xy)\}$, that arises from unknown instruments and an unknown bipartite process matrix, one certifies that this process matrix is causally nonseparable in a device-independent way if,  for some $a,b,x,y$,
\beq\label{eqDIcert}
p^Q(ab|xy)\neq\tr\left[(A_{a|x}\otimes B_{b|y})W^\text{sep}\right]
\eeq
for all causally separable process matrices $W^\text{sep}$ and all general instruments $\{A_{a|x}\}$ and $\{B_{b|y}\}$.
\end{definition}

\begin{definition}[Semi-device-independent certification]\label{defSDIcert}
Given a process behaviour $\{p^Q(ab|xy)\}$, that arises from unknown instruments on Alice's side, known instruments $\{\overline{B}_{b|y}\}$ on Bob's side, and an unknown bipartite process matrix, one certifies that this process matrix is causally nonseparable in a semi-device-independent way if, for some $a,b,x,y$,
\beq\label{eqSDIcert}
p^Q(ab|x \,, \overline{B}_{b|y})\neq\tr\left[(A_{a|x}\otimes \overline{B}_{b|y})W^\text{sep}\right]
\eeq
for all causally separable process matrices $W^\text{sep}$ and all general instruments $\{A_{a|x}\}$.
\end{definition}

On eqs. \eqref{eqDDcert}, \eqref{eqDIcert}, and \eqref{eqSDIcert} one can identify the quantities that are given -- the behaviour and the trusted instruments that belong to the device-dependent parties -- and the variables in the certification problem -- the unknown instruments that belong to the device-independent parties and any causally separable process matrix. If one can guarantee that, for any sets of instruments for the device-independent parties and any causally separable process matrix, the given behaviour cannot be described by the left-hand side of eqs. \eqref{eqDDcert}, \eqref{eqDIcert}, or \eqref{eqSDIcert}, then the fact that the process matrix that generated this behaviour is causally nonseparable is certified. A summary of the given quantities and variables in each certification scenario is provided by \cref{table:fixvar}.

Before proceeding we remark an analogy with the entanglement certification problem in which behaviours are assumed to arise from quantum measurements performed on a quantum state. In the entanglement certification case, device-dependent scenarios are related to entanglement witnesses \cite{horodecki96}, device-independent scenarios to Bell nonlocality \cite{brunner14}, and the semi-device-independent ones to EPR-steering \cite{wiseman07}.

\begin{table}
\begin{center}
{\renewcommand{\arraystretch}{1.5}
\begin{tabular}{ c | c }               
        	    	\multicolumn{1}{l}{\textbf{Device-dependent}} \\ 
	        \hline \hline            
		\small Given quantities &  \small  Variables \\              
		\hline 
	\small 	$\{p^Q(ab|xy)\}$ & \small  $W$ \\ 
		\small $\{\overline{A}_{a|x}\}$, $\{\overline{B}_{b|y}\}$& \\
           	\hline 
		\multicolumn{1}{c}{} \\ 
		\multicolumn{1}{l}{\textbf{Device-independent}} \\ 
	        \hline \hline            
		\small  Given quantities &\small  Variables \\              
		\hline 
		\small 	$\{p^Q(ab|xy)\}$ & \footnotesize $d_{A_I}$, $d_{A_O}$, $d_{B_I}$, $d_{B_O}$ \\
		& \small  $\{A_{a|x}\}$, $\{B_{b|y}\}$ \\
		& \small  $W$ \\
           	\hline 
		\multicolumn{1}{c}{} \\ 
		\multicolumn{1}{l}{\textbf{Semi-device-independent}} \\ 
	        \hline \hline            
		\small  Given quantities & \small Variables \\            
		\hline 
		\small $\{p^Q(ab|xy)\}$ &\small  $d_{A_I}$, $d_{A_O}$ \\
		\small $\{\overline{B}_{b|y}\}$ & \small $\{A_{a|x}\}$ \\
		&\small  $W$ \\
           	\hline 
\end{tabular}
}
\end{center}
\caption{Comparison between the given (known) quantities and the variables (unknown quantities) in each scenario with different levels of assumptions about the operations of each party involved in the task of certifying noncausal properties of a process matrix.}
\label{table:fixvar}
\end{table}
\normalsize


\subsection*{Device-dependent} 

We start be analysing the device-dependent scenario and \cref{defDDcert}. This scenario has been thoroughly studied before using the concept of causal witnesses \cite{araujo15}, analogous to entanglement witnesses \cite{horodecki96}, and now we formulate it under our certification paradigm. 

In this scenario, the behaviour and sets of instruments of both parties are given and we aim to check whether the given process behaviour is consistent with performing these exact instruments on a causally separable process matrix. This problem can be solved by semidefinite programming (SDP) with the following formulation:

 \footnotesize
\begin{align}
\begin{split}\label{sdpDD}
\text{given}                                   &\hspace{0.2cm} \{p^Q(ab|xy)\}, \{\overline{A}_{a|x}\}, \{\overline{B}_{b|y}\} \\
\text{find}	 				  &\hspace{0.2cm} W \\
\text{subject to}				  &\hspace{0.2cm} p^Q(ab|xy)=\tr\left[(\overline{A}_{a|x}\otimes \overline{B}_{b|y})W\right] \  \forall\,a,b,x,y \\
	 					  &\hspace{0.2cm} W \in \text{SEP}, \\
\end{split}
\end{align}
\normalsize
where SEP denotes the set of causally separable matrices (\textit{i.e.} $W$ is constrained to \cref{eqcausalW}) which can be characterized by SDP.

If the problem is infeasible, that is, if there does not exist a process matrix $W$ that satisfies the constraints of \cref{sdpDD}, then the process matrix that generated $\{p^Q(ab|xy)\}$ is certainly causally nonseparable. Consequently, there exists a causal witness that can certify it, without the need of performing full tomography of the process matrix \cite{araujo15}. On the contrary, if the problem is feasible, then the solution provides a causally separable process matrix that could have generated the given behaviour.
 
We now show that all causally nonseparable process matrices can be certified to be so in a device-dependent way.

\begin{theorem}\label{thmdd}
All causally nonseparable process matrices can be certified in a device-dependent way for some choice of instruments.
\end{theorem}

The proof of the above theorem follows from the fact that one can always consider a scenario where Alice and Bob have access to tomographically complete instruments, which allows for the complete characterization of the process matrix \cite{prugovecki77,caves02}. 


\subsection*{Device-independent} 

A device-independent approach for the process matrix formalism has been previously studied in terms of causal inequalities \cite{branciard16}, analogous to Bell inequalities \cite{brunner14}, and now we formulate it under our certification paradigm, exploring \cref{defDIcert}.

In a device-independent scenario, besides assuming that the given behaviour is a process behaviour following \cref{eq:processbehaviour}, no extra assumptions are made.
It is then necessary to check whether the given behaviour is consistent with probability distributions that come from any tensor product of pair of sets of instruments, with fixed number of inputs and outputs, performed on any causally separable process matrix of any dimension. 

Before characterizing the problem of whether a certain behaviour can be obtained by a causally separable process matrix, we ask the more fundamental question of whether any behaviour can be obtained by a general process matrix on which instruments are performed locally. That is, whether all general (valid) behaviours are process behaviours.

By definition, every process matrix leads to valid general behaviours. However, it is shown in ref.~\cite{costantino} that in the scenario where all parties have dichotomic inputs and outputs, the deterministic two-way signalling behaviour defined by $p^{2\text{WS}}(ab|xy):=\delta_{a,y}\delta_{b,x}$, where $\delta_{i,j}=1$ if $i=j$ and $\delta_{i,j}=0$ otherwise, cannot be obtained exactly by any process matrix. Here we show that one cannot obtain this two-way signalling behaviour even approximately for finite-dimensional process matrices. The proof can be found in \cref{app:robust_proof}.

\begin{theorem}\label{thmprobtwoway}
All process behaviours are valid behaviours, however, not all valid behaviours are process behaviours.

In particular, in the scenario where all parties have dichotomic inputs and outputs, any behaviour $\{p(ab|xy)\}$ such that 
$$\frac{1}{4}\sum_{a,b,x,y} \delta_{a,y}\delta_{b,x} \ p(ab|xy) > 1-\frac{1}{d+1}$$ 
is not a process behaviour for process matrices with total dimension $d_{A_I}d_{A_O}d_{B_I}d_{B_O} = d$.
\end{theorem}

Here we can make a parallel with Bell nonlocality, where the Popescu-Rohrlich behaviour is known to respect the non-signalling conditions which arise naturally in Bell scenarios but cannot be obtained by performing local measurements on entangled states \cite{tsirelson85,rastall85,popescu94}.

As for causal behaviours, the analogous question is also pertinent. Can all causal behaviours be obtained by pair of sets of instruments and causally separable process matrices? We answer this question positively, which allows us to relate the properties of a behaviour directly to the properties of the process matrices that could have given rise to it.

\begin{lemma}\label{lmmbehaviour}
A general behaviour is causal if and only if it is a process behaviour that can be obtained by a causally separable process matrix.
\end{lemma}
The proof is made by explicitly constructing the instruments and causally separable process matrix that can recover any causal behaviour. It can be found in \cref{app:quantumrealization}.

This result allows us to identify which causally nonseparable process matrices can be certified in a device-independent way:

\begin{theorem}\label{thmdi}
A process matrix can be certified to be causally nonseparable in a device-independent way if and only if it can generate a noncausal behaviour for some choice of instruments for Alice and Bob.
\end{theorem}

\begin{proof}
If a process matrix is causally separable then its behaviours will be causal. If a behaviour is causal, even though it could in principle have been generated by a causally nonseparable process matrix, according to \cref{lmmbehaviour} it can always be reproduced by a causally separable process matrix. Hence, no causal properties of the process matrix can be inferred. 
\end{proof}

From the formulation of the device-independent certification problem in \cref{defDIcert}, it is not clear whether one could obtain a simple characterisation to solve it. In particular, because there are no constraints on the dimension of the linear spaces and there is a product of variables (that represent the unknown instruments and process matrix). Interestingly, we can explore the above theorem to present simple necessary and sufficient conditions for a general behaviour to allow for device-independent certification of indefinite causal order. This follows from the fact that a behaviour can be checked to be noncausal by linear programming \cite{branciard16}. More explicitly, the certification problem can be formulated as follows:

\footnotesize
\begin{align}
\begin{split}\label{sdpDI}
\text{given}                                   &\hspace{0.2cm} \{p^Q(ab|xy)\} \\
\text{find}	 				  &\hspace{0.2cm} q_1(\lambda), q_2(\lambda) \\
\text{s.t.}					  &\hspace{0.2cm} p^Q(ab|xy)=\sum_\lambda \Big[ q_1(\lambda) D^{A \prec B}_\lambda(ab|xy) + \\ 
						  &\phantom{\hspace{2.5cm} } + q_2(\lambda) D^{B \prec A}_\lambda(ab|xy)\Big], \forall \ a,b,x,y \\
	 					  &\hspace{0.2cm} q_1(\lambda)\geq 0, \; q_2(\lambda)\geq 0, \ \forall \ \lambda,
\end{split}
\end{align}
\normalsize
where $\{D^{A \prec B}_\lambda(ab|xy) \}$ and  $\{D^{B \prec A}_\lambda(ab|xy) \}$ are the finite set of deterministic causal distributions described in ref.~\cite{branciard16}.

If the problem is infeasible, then the process matrix that was used to generate the process behaviour $\{p^Q(ab|xy)\}$ is certainly causally nonseparable and there exists a causal inequality that can witness it \cite{branciard16}. If the problem is feasible, then one can use the results presented in \cref{app:quantumrealization} to explicitly find a causally separable process matrix $W^\text{sep}$ and sets of instruments $\{A_{a|x}\}$ and $\{B_{b|y}\}$ such that $p^Q(ab|xy)=\tr\left[(A_{a|x}\otimes B_{b|y})W^\text{sep}\right]$.

Differently from the device-dependent scenario, it is known that some causally nonseparable process matrices cannot be certified in a device-independent way \cite{araujo15,oreshkov16,feix16}. In particular, there exist causally nonseparable bipartite process matrices that, for any choice of instruments of Alice and Bob, will always lead to causal behaviours. This result was first presented in ref.~\cite{feix16} and we rephrase it here:

\begin{proposition}[Device-dependent certifiable, device-independent noncertifiable process matrix]\label{prop1}
There exist causally nonseparable process matrices that, for any sets of instruments, always give rise to causal behaviours. That is, a causally nonseparable process matrix that cannot be certified in a device-independent way.

In particular, let $W\in\set{L}(\set{H}^{A_IA_OB_IB_O})$ be a process matrix and 
$W^{T_B}$ be the partial transposition of $W$ with respect to some basis in $\set{L}(\set{H}^{B_IB_O})$ for Bob.
If 
$W^{T_B}$ is causally separable, the behaviour generated by $p^Q(ab|xy)=\tr\left[(A_{a|x}\otimes B_{b|y})\,W\right]$ is causal for every sets of instruments $\{A_{a|x}\}$ and $\{B_{b|y}\}$.
\end{proposition}

We would like to remark that this phenomenon can be seen as a consequence of the choice of definition of causally separable process matrices. Recalling \cref{sec:preliminaries}, a causally ordered process matrix is defined as the most general operator $W^{A\prec B}$ that takes {any} pairs of sets of instruments to causally ordered behaviours according to $p^{A\prec B}(ab|xy)=\tr[(A_{a|x}\otimes B_{b|y})\ W^{A\prec B}]$. On the other hand, the definition of a causally separable process matrix $W^\text{sep}$ as a convex combination of process matrices with definite causal orders, instead of focusing on the behaviours, has an arguably more physical motivation of a classical mixture of causal orders. If the definition were to, alternatively, focus on the behaviours, then {a} natural choice would be to define a `causally separable' process matrix $\tilde{W}^\text{sep}$ as the most general operator that takes {any} pairs of sets of instruments to causal behaviours according to $p^\text{causal}(ab|xy)=\tr[(A_{a|x}\otimes B_{b|y})\ \tilde{W}^\text{sep}]$. With this alternative, {inequivalent} definition, `causally nonseparable' process matrices would {always lead to noncausal behaviours, for some choice of instruments}, by definition. We observe {the} phenomenon {of causally nonseparable process matrices leading exclusively to causal behaviours}, presented in \cref{prop1}, when we take the physically motivated definition, and it exposes an intrinsic difference between these two kinds of reasoning. In the next sections, we show how this phenomenon manifests itself in the semi-device-independent scenario.


\subsection*{Semi-device-independent} 

In this final scenario, which has not been explored for process matrices before, we have the information of the behaviour and the instruments of one party, in this case, Bob. According to \cref{defSDIcert}, one needs to check whether the given behaviour can be reproduced by performing these exact given instruments for Bob, and any set of instruments for Alice with fixed number of inputs and outputs, on any causally separable process matrix that has a fixed dimension on Bob's side. Since both the process matrix and Alice's instruments are variables in \cref{eqSDIcert}, it is not clear whether this problem can be solved by SDP.

{Our approach contrasts a previous one which exploits communication complexity tasks to certify indefinite causal order in process matrices assuming an upper bound for communication capacity between parties and the dimension of their local systems} \cite{feix15,guerin16}.

However, consider the following expression for a process behaviour,
\begin{align}
p(ab|xy) &= \tr\left[(A_{a|x}\otimes B_{b|y})W\right] \label{eqa} \\
&= \tr\left[B_{b|y}\,\tr_A(A_{a|x}\otimes\1^B\,W)\right] \label{eqb} \\
&=\tr\left(B_{b|y} w^Q_{a|x}\right), \ \ \ \forall \ a,b,x,y, \label{eqc}
\end{align}
which motivates us to define $w^Q_{a|x}$.

\begin{definition}[Process assemblage]
A process assemblage $\{w^Q_{a|x}\}$ is a set of operators $w^Q_{a|x}\in\set{L}(\set{H}^{B_I}\otimes\set{H}^{B_O})$ for which there exist a process matrix $W^{A_IA_OB_IB_O}$ and a set of instruments $\{A^{A_IA_O}_{a|x}\}$ such that
\beq
w^{Q}_{a|x} = \tr_{A_IA_O}\left[(A^{A_IA_O}_{a|x}\otimes\1^{B_IB_O})W^{A_IA_OB_IB_O}\right], \label{eqAQ}
\eeq
for all $a,x$.
\end{definition}

By defining the process assemblage, we gather all the variables in the certification problem in one object and can start to relate properties of this object to properties of the process matrix. We remark that the process assemblage generalizes the notion of assemblage in EPR-steering \cite{wiseman07,pusey13}, which is recovered when both Alice's and Bob's output spaces have $d_{A_O}=d_{B_O}=1$. Consequently, $\{A_{a|x}\}$ becomes a set of POVMs, $W$ becomes a bipartite quantum state, and the process assemblage recovers the steering assemblage $\sigma_{a|x}=\tr_A(A_{a|x}\otimes \1^B \rho^{AB})$ \cite{pusey13}.
 
Let us first examine the equation below more closely:
\beq\label{eqmostgenass}
p(ab|xy) = \tr\left(B_{b|y} w_{a|x}\right).
\eeq

In the same way that a process matrix was defined as the most general operator that takes sets of local instruments to a behaviour, we can define a \textit{general assemblage} to be the most general object that takes a set of instruments to a valid behaviour and respects linearity. In the \cref{app:assemblages}, we prove that this definition is equivalent to:
\begin{definition}[General assemblage]\label{defbigenass}
A general assemblage $\{w_{a|x}\}$ is a set of operators $w_{a|x}\in\set{L}(\set{H}^{B_I}\otimes\set{H}^{B_O})$ that satisfies
\begin{align}
w_{a|x} &\geq 0 \ \ \ \forall \ a,x \\
\tr \sum_a w_{a|x} &=d_{B_O} \ \ \ \forall \ x \\
\sum_a w_{a|x} &= _{B_O} \sum_a w_{a|x} \ \ \ \forall \ x.
\end{align}
\end{definition}

By defining the general assemblage as the most general set of operators that takes a set of instruments to a behaviour and respects linearity, we are no longer considering its relation with a process matrix or requiring that it is a process assemblage. 

If one compares the set of all general assemblages to the set of all process assemblages, it is clear that the set of general assemblages contains the set of process assemblages, since one can see from \cref{eqb} that all process assemblages lead to valid behaviours. But the former set is in principle larger, an outer approximation with a simpler characterisation. We show that, indeed, the set of general assemblages is larger than the set of process assemblages, because just like general behaviours, not all general assemblages can be realised by process matrices. 

\begin{theorem}\label{thmgenass}
All process assemblages are valid assemblages, however, not all valid assemblages are process assemblages.

In particular, in the scenario where Alice has dichotomic inputs and outputs, the general assemblage $\{w_{a|x}\}$ given by $w_{a|x}=\ketbra{x}{x}\otimes\ketbra{a}{a}$ is not a process assemblage.
\end{theorem}

The proof presented in \cref{app:robust_proof} is based on the fact that this assemblage can lead to a deterministic two-way signalling behaviour, which we know not to be attainable by process matrices from \cref{thmprobtwoway}.

Although the process assemblage can be regarded as a generalisation of the steering assemblage that arises in EPR-steering scenarios, here we point out an important difference between semi-device-independent certification of indefinite causal order and entanglement. A fundamental result on EPR-steering theory is that all steering assemblages admit a quantum realisation by performing POVMs on a quantum state \cite{schrodinger35,gisin89,hughston93}. On the other hand, \cref{thmgenass} shows that some process assemblages do not admit a quantum realisation by performing a set of instruments on a process matrix.


The next step is to assign causal properties to assemblages. A natural approach is to define that an assemblage $\{w^{Q,\, A\prec B}_{a|x}\}$ is causally ordered from Alice to Bob if it is a process assemblage that can be obtained from a process matrix that is causally ordered from Alice to Bob, namely, 
\beq\label{eqsjdn}
w^{Q,\,A\prec B}_{a|x} = \tr_A[(A_{a|x}\otimes \1^B)\,W^{A\prec B}] \ \ \ \forall \ a,x,
\eeq
for some set of instruments $\{A_{a|x}\}$ and some causally ordered process matrix $W^{A\prec B}$, and equivalently from Bob to Alice.

Since the above definition depends on an unknown set of instruments $\{A_{a|x}\}$, it is not easy to check whether a given general assemblage $\{w_{a|x}\}$ is causally ordered. We therefore derive a simpler characterisation of causally ordered assemblages that is equivalent to the one above. 

Analogously to how we defined a general assemblage, we characterise the most general set of operators $\{w^{A\prec B}_{a|x}\}$ that give rise to a causally ordered behaviour $\{p^{A\prec B}(ab|xy)\}$ according to the equation
\beq\label{eqmostgencauass}
p^{A\prec B}(ab|xy)=\tr(B_{b|y} w^{A\prec B}_{a|x}),
\eeq
for any set of instruments $\{B_{b|y}\}$ for Bob. We then prove its equivalence to the definition below in \cref{app:assemblages}.

\begin{definition}[Causally ordered assemblages]\label{defdeforderass}
An assemblage $\{w^{A\prec B}_{a|x}\}$ is causally ordered from Alice to Bob if it satisfies
\beq
w^{A\prec B}_{a|x}= _{B_O}w^{A\prec B}_{a|x} \ \ \ \forall \ a,x,
\eeq
while an assemblage $\{w^{B\prec A}_{a|x}\}$ is causally ordered from Bob to Alice if it satisfies
\beq
\sum_a w^{B\prec A}_{a|x}= \sum_a w^{B\prec A}_{a|x'} \ \ \ \forall \ x,x'.
\eeq
\end{definition}

In \cref{app:quantumrealization}, we show that all causally ordered assemblages $\{w^{A\prec B}_{a|x}\}$ and $\{w^{B\prec A}_{a|x}\}$ can be realized by some set of instruments $\{A_{a|x}\}$ and some causally ordered process matrix $W^{A\prec B}$ and $W^{B\prec A}$, respectively. That is, we show that all $\{w^{A\prec B}_{a|x}\}$ satisfy \cref{eqsjdn}, and analogously for the causal order $B\prec A$.

We now contrast the statement made in the previous paragraph with general assemblages. As stated before, the technique of defining the general assemblage as the most general set of linear operators that takes instruments to general behaviours results in an object that cannot always be described by process matrices. On the other hand, in the case of causally ordered assemblages, this technique yielded an object that can always be described by (causal) process matrices. The main point to be taken here is that to characterize the most general set of linear operators that takes instruments to some kind of behaviour is a mathematical artifice to find an outer approximation to the set of assemblages that are described by process matrices. The goal is to find an approximation of this set with a potentially simpler characterization. This approximation may be tight, as in the case of causal assemblages, or may not be tight, as in the case of general assemblages. We explore this further in \cref{app:tripartite} for assemblages in tripartite scenarios.

We now define a causal assemblage by taking the elements of the convex hull of causally ordered assemblages. 

\begin{definition}[Causal assemblage]\label{defcausalass}
An assemblage $\{w^\text{causal}_{a|x}\}$ is causal if it can be expressed as a convex combination of causally ordered assemblages, \textit{i.e.},
\beq\label{eqcausalA}
w^\text{causal}_{a|x} \coloneqq q w^{A\prec B}_{a|x} + (1-q)w^{B\prec A}_{a|x},
\eeq
for all $a,x$, where $0\leq q\leq 1$ is a real number. An assemblage that does not satisfy \cref{eqcausalA} is called a noncausal assemblage.
\end{definition}

We can now express our result in terms of the following lemma, proved in \cref{app:quantumrealization}.

\begin{lemma}\label{lmmass}
A general assemblage is causal if and only if it is a process assemblage that can be obtained from a causally separable process matrix.
\end{lemma}

This result allows us to identify which causally nonseparable process matrices can be certified in a semi-device-independent way: 

\begin{theorem}\label{thmsdi}
A process matrix is certified to be causally nonseparable in a semi-device-independent way if and only if it can generate a noncausal assemblage for some choice of instruments for Alice.
\end{theorem}

\begin{proof}
If a process matrix is causally separable then its assemblages will be causal. If the assemblage is causal, even though it could in principle have been generated by a causally nonseparable process matrix, according to \cref{lmmass} it can always be reproduced by a causally separable process matrix and hence this property cannot be certified. 
\end{proof}

{Note that all} requirements for an assemblage to be causal are linear and positive semidefinite constraints, hence one can check whether an assemblage is causal via SDP.

However, in our semi-device-independent certification scenario, the only information available is the process behaviour and Bob's instruments, not the assemblage itself. If it were the case that Bob's instruments are tomographically complete, he could obtain full information about the assemblage, and check whether it is causal via SDP. Nevertheless, we show that it is possible to check whether a given behaviour can certify indefinite causal order in a semi-device-independent scenario using SDP even without the knowledge of the assemblage. We do this by rephrasing our certification task in terms of an unknown assemblage.
\\

\noindent\textbf{Definition 3'} (Semi-device-independent certification, with assemblages)\textbf{.}
\textit{Given a behaviour $\{p^Q(ab|xy)\}$, that arises from unknown instruments on Alice's side, known instruments $\{\overline{B}_{b|y}\}$ on Bob's side, and an unknown bipartite process matrix, one certifies that this process matrix is causally nonseparable in a semi-device-independent way if, for some $a,b,x,y$,
\beq\label{eqSDIcert2}
p^Q(ab|x \,, \overline{B}_{b|y})\neq\tr(\overline{B}_{b|y}\,w^\text{causal}_{a|x})
\eeq
for all causal assemblages $\{w^{\text{causal}}_{a|x}\}$.}
\\

Now we are able to formulate the semi-device-independent certification problem in terms of SDP:

\begin{align}
\begin{split}\label{sdpSDI}
\text{given}                                   &\hspace{0.2cm} \{p^Q(ab|xy)\}, \{\overline{B}_{b|y}\} \\
\text{find}	 				  &\hspace{0.2cm} \{w_{a|x}\} \\
\text{s.t.}					  &\hspace{0.2cm} p^Q(ab|xy)=\tr(\overline{B}_{b|y}\,w_{a|x}) \ \forall\,a,b,x,y  \\
	 					  &\hspace{0.2cm} \{w_{a|x}\}\in\text{CAUSAL}, \\
\end{split}
\end{align}
where CAUSAL denotes the set of causal assemblages, that is, $\{w_{a|x}\}$ is constrained to \cref{eqcausalA}.

As in the previous cases, if the problem is infeasible, then the process matrix that was used to generate the process behaviour $\{p^Q(ab|xy)\}$ is certainly causally nonseparable.
If the problem is feasible, then one can use the results presented in \cref{app:quantumrealization} to explicitly find a causally separable process matrix $W^\text{sep}$ and sets of instruments $\{A_{a|x}\}$ such that $w_{a|x}=\tr_A[(A_{a|x}\otimes\1^B)\,W^\text{sep}]$.

All three SDP formulations we presented in \cref{sdpDD,sdpDI,sdpSDI} are feasibility problems which can be turned into optimisation problems that allow for a robust certification of indefinite causal order. We discuss this further in \cref{sec:switch}.

We now show that not all process matrices can be certified in a semi-device-independent way, as some process matrices cannot lead to noncausal assemblages. The proof is in \cref{app:counterexemples}.

\begin{theorem}[Device-dependent certifiable, semi-device-independent noncertifiable process matrix]\label{thmcounterexample}
There exist causally nonseparable process matrices that, for any sets of instruments on Alice's side, always give rise to causal assemblages. That is, causally nonseparable process matrices that cannot be certified in a semi-device-independent way.

In particular, let $W\in\set{L}(\set{H}^{A_IA_OB_IB_O})$ be a process matrix and $W^{T_A}$ be the partial transposition of $W$ with respect to some basis in $\set{L}(\set{H}^{A_IA_O})$ for Alice. If $W^{T_A}$ is causally separable, the assemblages generated by $w_{a|x}=\tr_A[(A_{a|x}\otimes\1^B)\,W]$ are causal for every set of instruments $\{A_{a|x}\}$.
\end{theorem}

We remark that the above theorem strictly extends \cref{prop1} which was first proved in ref.~\cite{feix16}. That is, with the same hypothesis -- that $W^{T_A}$ is causally separable -- we can make a stronger claim -- that $W$ cannot be certified as causally nonseparable even if Bob is treated in a device-dependent way (is trusted). 

In ref.~\cite{feix16}, the authors show that, when extended with an entangled state, the resulting process matrix can violate a causal inequality and can therefore be certified in a device-independent way. This implies that this extended process matrix can also be certified in a semi-device-independent way. Ref.~\cite{feix16} leaves as an open question the existence of a causally nonseparable bipartite process matrix that cannot be certified in a device-independent way even when extended by entanglement. We remark that this open question is also relevant in the context of semi-device-independent certification.

Another natural question also emerges: is there a bipartite process matrix that can be certified to be causally nonseparable in a semi-device-independent scenario but that cannot be certified to be causally nonseparable in any device-independent scenario? Although we believe such process matrix exists, no example is currently known.


\section{The Quantum Switch}\label{sec:switch}

The concepts of certification presented in the previous section have a natural generalisation to different multipartite scenarios. We are now going to illustrate one {particular} tripartite case by discussing and presenting some results involving the quantum switch \cite{chiribella12,chiribella13}. In \cref{app:tripartite}, {we present a detailed extension of the concepts and results from the bipartite case, introduced in} \cref{sec:certification}, {to the tripartite case in which the quantum switch is defined, and pave the way to future more general tripartite and multipartite extensions.}

On its first appearance, the \textit{quantum switch} was defined as a higher-order transformation that maps quantum channels into quantum channels and it can be defined as the following.  Let $U_A$ and $U_B$ be two unitary operators that act on the same space of a target state $\ket{\psi}^t$. Let $\ket{c}^{c}:=\alpha\ket{0}+\beta\ket{1}$, $|\alpha|^2+|\beta|^2=1$, be a `control' state that is able to coherently control the order in which the operations $U_A$ and $U_B$ are applied. The quantum switch acts as following:
\small
\beq\label{switch}
\text{switch}(U_A,U_B) = \ketbra{0}{0}^{c}\otimes U_A\,U_B + \ketbra{1}{1}^{c}\otimes U_B\,U_A.
\eeq	 
\normalsize
When applied to the state $\ket{c}^c\otimes\ket{\psi}^t$ we have
\small
\begin{align}\label{switch2}
\begin{split}
\text{switch}(U_A,U_B)\ket{c}\otimes\ket{\psi} = &\ \alpha\,\ket{0}\otimes U_A\,U_B\,\ket{\psi} \\
&+ \beta\,\ket{1}\otimes U_B\,U_A\,\ket{\psi}.
\end{split}
\end{align}
\normalsize

Physically, the equation above can be understood as the control qubit determining which unitary is going to be applied first on the target state $\ket{\psi}$. If the control qubit is in the state $\ket{0}$ ($\alpha=1,\beta=0$), the unitary $U_B$ is performed before the unitary $U_A$. If the control qubit is in the state $\ket{1}$ ($\alpha=0,\beta=1$), the unitary $U_B$ is performed before the unitary $U_A$. In general, if the control qubit is in the state $\ket{c}=\alpha\ket{0}+\beta\ket{1}$, $\alpha\neq0,\beta\neq0$, the output state will be in a coherent superposition of two different causal orders.

In the process matrix formalism, ref.~\cite{araujo17} has analysed the quantum switch as a four-partite process matrix of which the first party has only an output space, which defines the input target and control states, one party inputs $U_A$, another party inputs $U_B$, and a final party obtains the output control and target states. Here we follow instead the steps of ref.~\cite{araujo15} to associate \textit{tripartite process matrices} to the quantum switch. This can be done by absorbing the input target and control states into the process matrix and setting the third party with output space of dimension equal to one. Hence, the quantum switch is described by a family of tripartite process matrices that is shared among three parties, Alice, Bob, and Charlie, for which Charlie is always in the future of Alice and Bob, and the causal order between Alice and Bob may or may not be well defined. A consequence of the fact that Charlie is last and his output space $\set{H}^{C_O}$ has $d=1$ is that the most general instrument Charlie can perform is a POVM.

Formally, we define a family of tripartite process matrices associated to the quantum switch according to:

\begin{definition}[Quantum switch processes]
Let $\ket{w(\psi,\alpha,\beta)}\in\set{H}^{A_IA_OB_IB_OC^t_IC^c_I}$ be
\small
\begin{align}
\begin{split}
\ket{w\,(\psi,\alpha,\beta)} =&\,\alpha\, \ket{\psi}^{A_I}\ket{\Phi^+}^{A_OB_I}\ket{\Phi^+}^{B_OC^t_I}\ket{0}^{C^c_I} \\  
&+  \beta\, \ket{\psi}^{B_I}\ket{\Phi^+}^{B_OA_I}\ket{\Phi^+}^{A_OC^t_I}\ket{1}^{C^c_I}, 
\end{split}
\end{align}
\normalsize
where $\ket{\psi}$ is a $d$-dimensional pure state, $\alpha,\beta$ are complex numbers such that $|\alpha|^2+|\beta|^2=1$, and $\ket{\Phi^+}=\sum_{i=1}^{d}\ket{ii}$ is a maximally entangled unnormalised bipartite qudit state, the Choi representation of the identity channel.

Then, the pure quantum switch processes are a family of tripartite process matrices given by
\beq\label{eqpureswitch}
W_\text{switch}(\psi,\alpha,\beta) = \ketbra{w\,(\psi,\alpha,\beta)}{w\,(\psi,\alpha,\beta)}.
\eeq
\end{definition}

When the control state is in a nontrivial superposition of $\ket{0}$ and $\ket{1}$, the quantum switch processes have been shown to have some interesting properties \cite{araujo15,oreshkov16}. They are causally nonseparable process matrices, meaning they cannot be expressed as a convex combination of tripartite process matrices with definite causal ordered between Alice and Bob, with Charlie in their common future. Namely, when $\alpha\neq0$ and $\beta\neq0$,
\small
\beq
W_\text{switch}(\psi,\alpha,\beta) \neq q W^{A\prec B\prec C} + (1-q) W^{B\prec A\prec C},
\eeq
\normalsize
for all $\ket{\psi}$, and all real numbers $0\leq q\leq 1$. The exact definitions of causally ordered and general tripartite process matrices can be found in \cref{app:tripartite}. 

However, when Charlie is traced out, the resulting bipartite process matrices shared by Alice and Bob are causally separable, namely,
\small
\beq
\tr_{C^c_IC^t_I}[W_\text{switch}(\psi,\alpha,\beta)] = q W^{A\prec B} + (1-q) W^{B\prec A},
\eeq
\normalsize
for all $\ket{\psi}$, $\alpha$, and $\beta$, {where $q =|\alpha^2|$}.

These causally nonseparable tripartite process matrices can be certified in a device-dependent way, since \cref{thmdd} also holds for tripartite process matrices. Yet, it has been shown in refs.~\cite{araujo15,oreshkov16} that the quantum switch processes cannot be certified in a device-independent way, as they always lead to causal behaviours for any choice of instruments of Alice, Bob, and Charlie. It remains to find out whether these processes can be certified in semi-device-independent scenarios.

For this purpose, we extend all concepts and methods from bipartite semi-device-independent certification. Much like in the bipartite case, we make different assumptions about the knowledge of the operations performed by each party. We call \textit{untrusted}~(U) a party that is treated in a device-independent way and \textit{trusted}~(T) a party that is treated in a device-dependent way, and we use the convention Alice Bob Charlie for denoting the parties. For example, a scenario TTU means Alice~=~T (device-dependent), Bob~=~T (device-dependent), and Charlie~=~U (device-independent). The four inequivalent semi-device-independent tripartite scenarios are TTU, TUU, UTT, and UUT. 

The core idea of the certification task remains the same. For a given process behaviour and given sets of instruments for the trusted parties, one needs to check whether it is possible that this behaviour comes from performing this instruments on a causally separable tripartite process matrix. We also derive the concepts of general, process, causally ordered, and causal assemblages for each scenario. In \cref{app:tripartite}, we provide all details and calculations, including for the tripartite device-dependent TTT and -independent UUU scenarios.

Our next theorem strengthens the previous result \cite{araujo15} that showed that the quantum switch processes cannot be certified in a full device-independent scenario, \textit{i.e.}, in the UUU scenario. We show that when the instruments of Alice and Bob are unknown, even if the measurements performed by Charlie are known, the quantum switch processes can never be proven to be causally nonseparable, for any pairs of sets of instruments for Alice and Bob. In other words, we prove that the quantum switch processes cannot be certified to be causally nonseparable in the UUT scenario. The previous result of the full device-independent scenario can now be seen as a particular case of the theorem we now present, whose proof is in \cref{app:switchUUTcausal}.

\begin{theorem}\label{thmswitch}
The quantum switch processes cannot be certified to be causally nonseparable on a semi-device-independent scenario where Alice and Bob are untrusted and Charlie is trusted (UUT).

Moreover, any tripartite process matrix $W\in\set{L}(\set{H}^{A_IA_OB_IB_OC_I})$, with Charlie in the future of Alice and Bob, that satisfies
\small
\begin{align}
\begin{split}
\tr[(A^{A_IA_O}_{a|x}&\otimes B^{B_IB_O}_{b|y}\otimes \1^{C_I})W^{A_IA_OB_IB_OC_I}] = \\
&q p^{A\prec B}(ab|xy) + (1-q) p^{B\prec A}(ab|xy),
\end{split}
\end{align}
\normalsize
for all $a,b,x,y$, where $0 \leq q \leq1$ is a real number, cannot be certified to be causally nonseparable in a UUT scenario.
\end{theorem}

We now show that in the three remaining semi-device-independent scenarios, TTU, TUU, and UTT, the quantum switch processes \textit{can} be certified to be causally nonseparable, proving that they can demonstrate stronger noncausal properties than it was previously known.

For our remaining calculations we use the \textit{reduced quantum switch process}
\beq\label{eqswitch}
W_\text{red} \coloneqq \tr_{C^t_I} \left[W_\text{switch}\left(\ket{0},\frac{1}{\sqrt{2}},\frac{1}{\sqrt{2}}\right)\right].
\eeq
By choosing a reduced, mixed switch process in these scenarios, in which the space of the target state is traced out, we guarantee that the pure switch process version can also be certified without performing any measurements on the output target space.

We show that the reduced quantum switch process $W_\text{red}$ can be certified in the TTU, TUU, and UTT scenarios by providing sets of instruments for the device-independent (untrusted) parties that, when applied to $W_\text{red}$, generate TTU-, TUU-, and UTT-assemblages that are noncausal. We prove these assemblages to be noncausal by means of SDP.

To calculate the assemblages, we choose the following instruments for each untrusted party:
\footnotesize
\begin{align}\label{eqswitchins}
\begin{split}
A^{A_IA_O}_{0|0} &= B^{B_IB_O}_{0|0} = \ketbra{0}{0}\otimes\ketbra{0}{0}, \ \ \ \ \ \ M^{C^c_I}_{0|0} = \ketbra{+}{+}, \\
A^{A_IA_O}_{1|0} &= B^{B_IB_O}_{1|0} = \ketbra{1}{1}\otimes\ketbra{1}{1}, \ \ \ \ \ \ M^{C^c_I}_{1|0} = \ketbra{-}{-}, \\
A^{A_IA_O}_{0|1} &= B^{B_IB_O}_{0|1} = \ketbra{+}{+}\otimes\ketbra{+}{+}, \\ 
A^{A_IA_O}_{1|1} &= B^{B_IB_O}_{1|1} = \ketbra{-}{-}\otimes\ketbra{-}{-}, 
\end{split}
\end{align}
\normalsize
where $\ket{\pm} = \frac{1}{\sqrt{2}}(\ket{0}\pm\ket{1})$.

Let us illustrate with the TUU case. To construct the TUU-assemblage, we use the instruments from \cref{eqswitchins} for the untrusted parties, Bob and Charlie, according to
\small
\beq
w^{\text{switch}}_{bc|yz} = \tr_{BC} [(\1^A\otimes B_{b|y}\otimes M_{c|z})\, W_\text{red}] \ \ \ \forall \ b,c,y,z.
\eeq
\normalsize
Then, we show via SDP that
\small
\beq\label{equnfea}
w^{\text{switch}}_{bc|yz} \neq q w^{A\prec B\prec C}_{bc|yz} + (1-q) w^{B\prec A\prec C}_{bc|yz}, \ \ \ \forall \ b,c,y,z,
\eeq
\normalsize
proving that the switch process can be certified in this scenario. This is possible due to the SDP characterization of causal assemblages presented in \cref{app:tripartite} (for instance, see \cref{defcausalTUUass} for the TUU scenario). It follows analogously for the scenarios TTU and UTT.

To be able to compare and quantify the causal properties of the quantum switch across different certification scenarios, from full device-dependent to full device-independent, we introduce a robustness measure.
We start by defining the noisy version of the switch, a mixture of the reduced quantum switch process with a `trivial' process matrix (the normalised identity):
\beq\label{eqnoise}
W^\eta_\text{red}\coloneqq\eta\,\frac{\1}{d_I} + (1-\eta)W_\text{red},
\eeq
where $d_I=d_{A_I}d_{B_I}d_{C^c_I}$ is the dimension of the joint input spaces.

We then estimate the minimum value of $\eta$ for which $W^\eta_\text{red}$ has only causal properties in a given scenario. For example, in the device-dependent scenario, this is the minimum value of $\eta$ for which $W^\eta_\text{red}$ is causally separable. In a semi-device-dependent scenario, this is the minimum value of $\eta$ for which $W^\eta_\text{red}$ only generates causal assemblages. Finally, for the device-independent scenario, this is the minimum value of $\eta$ for which $W^\eta_\text{red}$ only generates causal behaviours.

It is immediate to see that in the UUU (device-independent) scenario, $\eta^*=0$, since the switch process always generates causal behaviours \cite{araujo15}. Equivalently for the UUT scenario, as a direct consequence of our \cref{thmswitch}. For the TTT scenario, we evaluate via SDP a value of $\eta^{*}=0.6118$ for which $W^{\eta^*}_\text{red}$ is causally separable and bellow which it is causally nonseparable.

In the remaining scenarios, TTU, TUU, and UTT, in order to calculate the exact value of $\eta^*$, one should consider every possible assemblage that could be generated from $W^\eta_\text{red}$, by optimizing over the set of instruments of the device-independent (untrusted) parties. Since we only consider the fixed instruments of \cref{eqswitchins}, we calculate lower bounds for $\eta^*$. Additionally, as detailed in \cref{app:tripartite}, in some of these scenarios our SDP characterization of the set of causal assemblages only constitutes an outer approximation of the set of assemblages that can be described by causal process matrices. Since with SDP we calculate the minimum $\eta$ for the assemblages to be inside this outer approximation, this again gives only a lower bound for $\eta^*$.

Let us illustrate again with the TUU case. We construct a noisy TUU-assemblage using the instruments from \cref{eqswitchins} and $W^\eta_\text{red}$ according to 
\footnotesize
\beq
w^{\eta,\text{switch}}_{bc|yz} = \tr_{BC} [(\1^A\otimes B_{b|y}\otimes M_{c|z})\, W^\eta_\text{red}] \ \ \ \forall \ b,c,y,z.
\eeq
\normalsize
We then calculate via SDP the minimum value of $\eta$ for which $w^{\eta,\text{switch}}_{bc|yz}$ is causal and below which it is noncausal. This value constitutes a lower bound for $\eta^*$. In the TUU scenario, we evaluate $\eta^*\geq0.1621$ for the reduced switch process. Analogously, we evaluate $\eta^*\geq0.1802$ in the UTT scenario. Finally, in the TTU scenario, $\eta^*\geq0.5687$. All these values are summarized in \cref{table:switch}. {All code used to obtain these results is freely available in an online repository} \cite{github}. We remark that the set of instruments required to obtain a robust certification of noncausal separability of the quantum switch is relatively simple and that Charlie can be restricted to perform a single POVM.
\\

The quantum switch has motivated several experiments that explore optical interferometers to certify indefinite causal order of process matrices \cite{procopio15,rubino17,goswami18-1,goswami18-2}. Up to now, all experimental results rely on, among other assumptions, complete knowledge of the instruments to certify of causal nonseparability, \textit{i.e.}, they are fully device-dependent. As shown in this section, one can also certify the causal nonseparability of the quantum switch without trusting some of the instruments and measurement apparatuses, in a semi-device-inpedendent way. We have used the machinery developed in this paper to analyse the experiments of refs.~\cite{rubino17} and \cite{goswami18-1} and concluded that, the instruments used in these experiments could allow us to make a stronger claim than what was reported. More precisely, the instruments used to certify that the quantum switch is causally nonseparable on refs.~\cite{rubino17} and \cite{goswami18-1} can lead to a semi-device-independent certification of the noncausal properties of the quantum switch in the TTU scenario. We discuss this results further in \cref{app:experimental}. 

\begin{table}
\begin{center}
{\renewcommand{\arraystretch}{1.5}
\begin{tabular}{| c | c |}
		\hline
		\multicolumn{2}{|c|}{\textbf{TTT}}  \\ 
		\multicolumn{2}{|c|}{$\eta^*=0.6118$} \\ 
		\multicolumn{2}{|c|}{\noncausalc{\sc{Noncausal}}} \\
		\hline                             
		\phantom{\hspace{0.9cm}}\textbf{UTT}\phantom{\hspace{0.9cm}}	&	\phantom{\hspace{0.9cm}}\textbf{TTU}\phantom{\hspace{0.9cm}} \\
		$\eta^*\geq0.1802$	&	$\eta^*\geq0.5687$ \\
		\noncausalc{\sc{Noncausal}}		& 	\noncausalc{\sc{Noncausal}} \\
		\hline 
		\textbf{UUT}			&	\textbf{TUU} \\
		$\eta^*=0$		&	$\eta^*\geq0.1621$ \\
		\causalc{\sc{Causal}}		&	\noncausalc{\sc{Noncausal}} \\
		\hline
		\multicolumn{2}{|c|}{\textbf{UUU}} \\
		\multicolumn{2}{|c|}{$\eta^*=0$} \\ 
		\multicolumn{2}{|c|}{\causalc{\sc{Causal}}} \\
		\hline
\end{tabular}
}
\end{center}
\caption{The quantum switch can be certified to be causally nonseparable in scenarios TTT, UTT, TTU, and TUU and cannot be certified in scenarios UUT and UUU, where T stands for trusted (device-dependent), U for untrusted (device-independent) and we have chosen the order Alice, Bob, and Charlie (for instance, TTU represents the scenario where Alice and Bob are treated in a device-dependent and Charlie in a device-independent way). The values and bounds for $\eta^*$ concern the critical value of the mixing parameter $0\leq\eta\leq 1$ in \cref{eqnoise} for which the quantum switch cannot be certified to causally nonseparable on each scenario, and below which, it can be certified. All non-zero values were obtained via SDP. {All code is publicly available in an online repository} \cite{github}.}
\label{table:switch}
\end{table}
\normalsize


\section*{Conclusions}\label{sec:conclusions}

We developed a framework for certifying indefinite causal order in the process matrix formalism under different sets of assumptions about the operations of the involved parties. In particular, we constructed a semi-device-independent approach to certification of causally nonseparable process matrices, and unified previously explored device-dependent and -independent approaches. 

We showed that the sets of causally nonseparable process matrices that can be certified in each scenario are different. More specifically, we proved that some bipartite process matrices can be certified to be causally nonseparable in a device-dependent way but not in a semi-device-independent way, and that some tripartite process matrices can be certified to be causally nonseparable in a semi-device-independent way but not in a device-independent way. 

In our framework, we formulated the problem of certifying causally nonseparable process matrices in the device-dependent, semi-device-independent, and device-independent scenarios in terms of semidefinite programming (SDP), implying they can be efficiently solved. 

We also showed that some noncausal behaviours and some noncausal assemblages cannot be obtained by process matrices according to the rules of quantum mechanics. For the device-independent case, we presented non-trivial bounds that relate the dimension of a process matrix with its maximal attainable violation of a causal game inequality. Concerning bipartite causal behaviours and causal assemblages, we explicitly showed how to obtain them from causally separable process matrices.

Concerning the quantum switch, we proved that it can produce noncausal correlations, that is, can be certified to be causally nonseparable, in three out of four semi-device-independent scenarios, and proved that its noncausal properties cannot be certified in the remaining one. 

Finally, we showed that previous experiments that claim to have certified causal nonseparability with the quantum switch under device-dependent assumptions \cite{rubino17,goswami18-1} could have, in principle, dropped some assumptions to achieve a stronger form of certification. Our results provide the theoretical basis for a future experimental demonstration of stronger noncausal phenomena that will rely on weaker assumptions than previous ones. 
\\ 

\noindent\textit{Acknowledgements.} We are grateful to Cyril Branciard for insightful discussions, particularly regarding \cref{thmswitch}, and to Costantino Budroni, for sharing a preliminary version of ref.~\cite{costantino}. 

This work was supported by the Austrian Science Fund (FWF) through the START project Y879-N27, the Excellence Initiative of the German Federal and State Governments (Grant ZUK 81), the Japan Society for the Promotion of Science (JSPS) by KAKENHI grant No. 16F16769, and the Q-LEAP project of the MEXT, Japan. CB acknowledges the support of the FQXi and the Austrian Science Fund (FWF) through the SFB project ``BeyondC'', the projects I-2526-N27 and I-2906. This publication was made possible through the support of a grant from the John Templeton Foundation. The opinions expressed in this publication are those of the authors and do not necessarily reflect the views of the John Templeton Foundation. 



\begin{small}

\end{small}


\onecolumngrid 
\appendix
\pagebreak

\section*{{Appendix}}

In these appendices we provide material to complement the main text. In \cref{app:robust_proof}, we present the proofs of \cref{thmprobtwoway,thmgenass}, regarding behaviours and assemblages that cannot be expressed in terms of process matrices. In \cref{app:quantumrealization}, we present the proofs of \cref{lmmbehaviour,lmmass}, regarding the realization of causal behaviours and assemblages in terms of causally separable process matrices. In \cref{app:assemblages}, we show how to obtain the characterization of general and causal assemblages presented in \cref{defbigenass,defdeforderass}. In \cref{app:counterexemples}, we present the proof of \cref{thmcounterexample}, showing a class of causally nonseparable process matrices that cannot be certified in a semi-device-independent way. In \cref{app:tripartite}, we re-derive all concepts and results of the main text concerning certification of bipartite process matrices for {the tripartite process matrices whose third party is always in the future of the other two.} We start by defining all notions of certification for {these} tripartite process matrices and then explore each scenario (TTT, UUU, TTU, TUU, UTT, and UUT) in detail. This appendix contains technical results not used nor mentioned in the main text. In \cref{app:switchUUTcausal}, we present the proof of \cref{thmswitch}, \textit{i.e.}, we prove that the quantum switch cannot be certified to be causally nonseparable in the UUT scenario. Finally, in \cref{app:experimental}, we present our theoretical analysis of the quantum switch experiments of refs.~\cite{rubino17} and \cite{goswami18-1}.


\section{Behaviours and assemblages unattainable by process matrices}\label{app:robust_proof}

In this appendix, we prove \cref{thmprobtwoway} and \cref{thmgenass}, concerning behaviours and assemblages which cannot be obtained by process matrices. We start by presenting and proving the following lemma, which will be necessary for the proof of \cref{thmprobtwoway}.

\begin{lemma}\label{lema}
Let $A \in \mathcal{L}(\mathcal H^1 \otimes \mathcal H^2)$ be a positive semidefinite operator. It holds that
\begin{equation}\label{eq:positiveA2}
	d\, \tr_{2}(A) \otimes \id - A \geq 0.
\end{equation}
where $d = \min\{d_1,d_2\}$ and $ \id$ denotes the identity operator acting on $\mathcal{H}^2$.
\end{lemma}

\begin{proof}

Let us start with the case where $A$ is a rank-1 operator, so that $A = \proj{\lambda}$ for some unnormalized vector $\ket{\lambda}$. 
The Schmidt decomposition ensures that every bipartite vector $\ket{\lambda}\in(\mathcal H^1 \otimes \mathcal H^2)$  can be written as 
$\ket{\lambda}=\sum_{i=1}^d \lambda_i \ket{ii}$ for some real non-negative coefficients $\{\lambda_i\}_{i=1}^d$ and $d = \min\{d_1,d_2\}$.
We can then define a diagonal operator $D:\mathcal{H}^1\to\mathcal{H}^1$, $D:=\sqrt{d}\sum_{i=1}^d \lambda_i \ketbra{i}{i}$ such that $\ket{\lambda}=D\otimes \id \ket{\phi^+_d}$ where $\ket{\phi^+_d}:=\sum_{i=1}^d \frac{1}{\sqrt{d}} \ket{ii}$ is a $d$-dimensional maximally entangled state acting on $\mathcal{H}^1\otimes\mathcal{H}^2$.

Using the diagonal operator $D$, and $D^\dagger=D$, the partial trace $\tr_{2}(A)$ can be written as
\begin{align}
	\tr_{2}(A) &= \tr_2 \left( D\otimes \id \ketbra{\phi^+_d}{\phi^+_d}  D\otimes \id\right) \\
		   &=\frac{DD}{d}.
\end{align}
Direct calculation of the left-hand side of inequality \eqref{eq:positiveA2} leads to
\begin{align}
	d\, \tr_{2}(A) \otimes \id - A &=
	d\, \frac{DD\otimes \id}{d} - D \otimes \id \ketbra{\phi^+_d}{\phi^+_d}  D \otimes \id \\
		   &=D\otimes \id \left( \id \otimes \id - \ketbra{\phi^+_d}{\phi^+_d} \right) D \otimes \id \\
		   &\geq 0,
\end{align}
where the last inequality holds since $\id \otimes \id - \ketbra{\phi^+_d}{\phi^+_d}\geq0$ and $D\geq0$ .

To prove the general case, note that we can write $A$ as a sum of rank-1 operators, $A = \sum_i \proj{\lambda^{(i)}}$. Since for every $i$ it holds that
\begin{equation}
d\,\tr_{2}\left(\proj{\lambda^{(i)}}\right)\otimes \id - \proj{\lambda^{(i)}} \geq0 ,
\end{equation}
we also have that
\begin{equation}
d\,\tr_{2}\left(\sum_i \proj{\lambda^{(i)}}\right)\otimes \id - \sum_i \proj{\lambda^{(i)}} \geq0 ,
\end{equation}
and hence,
\begin{equation}\label{jumpstart}
d\,\tr_{2}(A)\otimes \id -  A \geq0.
\end{equation}
\end{proof}

\addtocounter{equation}{3}

\noindent Note: This document was amended to fix the above proof of lemma~\ref{lema}. To keep the equation numbering consistent with the originally published version (v3), we increment the equation counter at this point, jumping from eq.~\eqref{jumpstart} to eq.~\eqref{jumpend}.

\setcounter{theorem}{1}
\begin{theorem}
All process behaviours are valid behaviours, however, not all valid behaviours are process behaviours.

In particular, in the scenario where all parties have dichotomic inputs and outputs, any behaviour $\{p(ab|xy)\}$ such that 
$$\frac{1}{4}\sum_{a,b,x,y} \delta_{a,y}\delta_{b,x} \ p(ab|xy) > 1-\frac{1}{d+1}$$ 
is not a process behaviour for process matrices with total dimension $d_{A_I}d_{A_O}d_{B_I}d_{B_O} = d$.
\end{theorem}

\begin{proof}
It follows by definition that all process behaviours are valid behaviours, we now show that there exist valid behaviours that are not process behaviours for any finite dimension. Assume that there exists a process matrix $W$ with total dimension $d_{A_I}d_{A_O}d_{B_I}d_{B_O} = d$ and instruments $\{A_{a|x}\},\{B_{b|y}\}$ such that
\beq\label{jumpend}
p(ab|xy) = \tr\big[W(A_{a|x}\otimes B_{b|y})\big]
\eeq
and
\beq
p^\text{succ}_\text{GYNI} = \frac14\sum_{abxy} \delta_{ay}\delta_{bx} p(ab|xy) > 1-\frac1{d+1},
\eeq
where $\delta_{ij}$ is the Kronecker's delta function and  $p^\text{succ}_\text{GYNI}$ is the probability of success  achieved by this behaviour in the GYNI causal game (defined in ref.~\cite{branciard16}). 

Let then we define the GYNI operator
\beq
M := \frac14\sum_{abxy} \delta_{ay}\delta_{bx} A_{a|x}\otimes B_{b|y},
\eeq
so that 
\beq
p^\text{succ}_\text{GYNI} = \tr(WM).
\eeq
Since $W$ is a valid process matrix, it admits the decomposition
\beq
W = {}_{A_O}W + {}_{B_O}W - {}_{A_OB_O}W,
\eeq
 and by linearity we have 
\beq
p^\text{succ}_\text{GYNI} = \tr({}_{A_O}WM)+\tr({}_{B_O}WM)-\tr({}_{A_OB_O}WM).
\eeq

Since ${}_{A_O}W$ and ${}_{B_O}W$ are causally ordered process matrices, they cannot violate the GYNI causal inequality and respect \cite{branciard16}
\beq
	\tr(M W^{\text{sep}})\le \frac{1}{2}.
\eeq
Also, it follows from \cref{lema} that
\beq
	-\tr({}_{A_OB_O}WM) \leq  -\frac1d\tr(WM).
\eeq
We than have 
\begin{align}
p^\text{succ}_\text{GYNI} = \tr(WM) &= \tr({}_{A_O}WM)+\tr({}_{B_O}WM)-\tr({}_{A_OB_O}WM) \\
						       &\leq \frac{1}{2}  \quad \quad\quad\quad\;	+\frac{1}{2} \quad \quad\quad\quad\; - \frac1d\tr(WM),
\end{align}
which implies 
\beq
p^\text{succ}_\text{GYNI} = \tr(WM) \leq 1 - \frac{1}{d+1},
\eeq
and bounds the maximal attainable value for process behaviours on the GYNI causal game, which can attain $p^\text{succ}_\text{GYNI}=1 $ for general behaviours.
	
Note that this robust bound only holds for finite dimensional process matrices. In ref. \cite{costantino}, the authors have shown that a process matrix $W$ cannot attain the \textit{exact} maximal success probability in the GYNI causal game \textit{i.e.,} $p^\text{succ}_\text{GYNI} = \tr(WM)=1$, even with infinite dimension.
\end{proof}

\setcounter{theorem}{3}
\begin{theorem}
All process assemblages are valid assemblages, however, not all valid assemblages are process assemblages.

In particular, in the scenario where Alice has dichotomic inputs and outputs, the general assemblage $\{w_{a|x}\}$ given by $w_{a|x}=\ketbra{x}{x}\otimes\ketbra{a}{a}$ is not a process assemblage.
\end{theorem}

\begin{proof}
First we see that $w_{a|x}=\ketbra{x}{x}\otimes\ketbra{a}{a}$ is a valid assemblage, since $\ketbra{x}{x}\otimes\ketbra{a}{a}\geq 0$ for all $a,x$, and 
\beq
\sum_a w_{a|x} = \ketbra{x}{x} \otimes \1 \ \ \ \forall \ x.
\eeq

Assume then that $\{w_{a|x}\}$ is a process assemblage. Then, there must exist a process matrix $W$ and instruments $\{A_{a|x}\}$ such that
\beq
w_{a|x} = \tr_{A_IA_O}\De{ \de{A_{a|x} \otimes \id}W},
\eeq
and that for any set of instruments $\{B_{b|y}\}$,
\beq
p(ab|xy) = \tr(B_{b|y}w_{a|x}) = \tr\De{\de{A_{a|x} \otimes B_{b|y}}W }.
\eeq

Consider now the valid set of instruments given by
\beq
B_{b|y} =  \ketbra{b}{b} \otimes \ketbra{y}{y}.
\eeq
Behaviours obtained by this set of instruments and the assemblage $\{w_{a|x}\}$ would be
\begin{align}
p(ab|xy) &= \tr(B_{b|y}w_{a|x}) \\
&= \delta_{b,x}\delta_{a,y} \ \ \ \forall \ a,b,x,y,
\end{align}
which according to \cref{thmprobtwoway} cannot be obtained by any finite dimensional process matrix, contradicting the hypothesis that $\{w_{a|x}\}$ is a process assemblage.
\end{proof}


\section{Behaviours and assemblages attainable by causally separable process matrices}\label{app:quantumrealization}

In this appendix we prove \cref{lmmbehaviour} and \cref{lmmass}, which we reestate for the convenience of the reader.
\setcounter{lemma}{0}

\begin{lemma}\label{lmmbehaviour2}
A general behaviour is causal if and only if it is a process behaviour that can be obtained by a causally separable process matrix.
\end{lemma}

\begin{proof}
To prove the if part we show that causally separable process matrices can only give rise to causal behaviours, regardless of the instruments performed on them. 

We start by showing that a behaviour $\{p(ab|xy)\}$ that comes from acting with any sets of instruments $\{A_{a|x}\}$ and $\{B_{b|y}\}$ on a process matrix that is causally ordered from Alice to Bob $W^{A\prec B}$ is also causally ordered from Alice to Bob.
Given $\{p(ab|xy)\}$ that arises from $W^{A\prec B}$ according to
\begin{align}
p(ab|xy)&= \tr\left[(A_{a|x}\otimes B_{b|y})\ W^{A\prec B}\right] \\
&= \tr\Big[(A^{A_IA_O}_{a|x}\otimes B^{B_IB_O}_{b|y})\ W^{A_IA_OB_I}\otimes\frac{\1^{B_O}}{d_{B_O}}\Big],
\end{align}
one can check that Alice's marginal probability distributions,
\begin{align}
\sum_b p(ab|xy) &= \sum_b \tr\Big[(A^{A_IA_O}_{a|x}\otimes B^{B_IB_O}_{b|y})\ W^{A_IA_OB_I}\otimes\frac{\1^{B_O}}{d_{B_O}}\Big] \\
&=  \tr\Big[(A^{A_IA_O}_{a|x}\otimes \1^{B_I})(W^{A_IA_OB_I})\ \tr_{B_O}(\1^{A_IA_O}\otimes\frac{1}{d_{B_O}}\sum_bB^{B_IB_O}_{b|y})\Big] \\
&= \frac{1}{d_{B_O}} \tr\left[(A^{A_IA_O}_{a|x}\otimes \1^{B_I})W^{A_IA_OB_I}\right],
\end{align}
are independent of $y$ for all $a,x$. Hence, $p(ab|xy)=p^{A\prec B}(ab|xy)$ is causally ordered from Alice to Bob.
Equivalently, $W^{B\prec A}$ implies $\{p^{B\prec A}(ab|xy)\}$. 

Hence, behaviours that come from a causally separable process matrix according to
\begin{align}
p(ab|xy)&=\tr\left[(A_{a|x}\otimes B_{b|y})\ W^\text{causal}\right] \\
&= \tr\left[(A_{a|x}\otimes B_{b|y})\ (q W^{A\prec B} + (1-q)W^{B\prec A}) \right] \\
&= q p^{A\prec B}(ab|xy) + (1-q) p^{B\prec A}(ab|xy)
\end{align}
are causal by definition. In other words, causally separable process matrices can only generate causal behaviours.

To prove the only if part we show that all behaviours that are causal can be reproduced by performing some instruments on some causally separable process matrix. 

Given a causal behaviour $\{p^\text{causal}(ab|xy)\}$, according to Bayes' rule and the non-signalling principle, one can decompose it in the following form:
\begin{align}
p^\text{causal}(ab|xy) =& q \, p^{A\prec B}(ab|xy) + (1-q)p^{B\prec A}(ab|xy) \\
=& q \, p_A^{A\prec B}(a|x)p_B^{A\prec B}(b|axy) + (1-q)p_B^{B\prec A}(b|y)p_A^{B\prec A}(a|bxy),
\end{align}
so that $\{p_A^{A\prec B}(a|x)\}$, $\{p_B^{A\prec B}(b|axy)\}$, $\{p_B^{B\prec A}(b|y)\}$, $\{p_A^{B\prec A}(a|bxy)\}$, and $q$ are given quantities.  

First, for the contribution that is causally ordered from Alice to Bob, $p^{A\prec B}(ab|xy)$, we construct instruments $\{A^{A\prec B}_{a|x}\}$ and $\{B^{A\prec B}_{b|y}\}$ according to
\begin{align}
A^{A\prec B}_{a|x} &= \1^{A_I}\otimes p_A^{A\prec B}(a|x) \ketbra{ax}{ax}^{A_O} \label{ax} \\ 
B^{A\prec B}_{b|y} &= \sum_{a,x}p_B^{A\prec B}(b|axy)\ketbra{ax}{ax}^{B_I}\otimes\frac{\1^{B_O}}{d_{B_O}}, \	\label{by} 
\end{align}
and process matrix
\beq
W^{A\prec B} = \frac{\1^{A_I}}{d_{A_I}} \otimes \ketbra{\Phi^+}{\Phi^+}^{A_OB_I} \otimes \1^{B_O},
\eeq
where $\ketbra{\Phi^+}{\Phi^+}$ is the Choi operator of the identity channel, with $\ket{\Phi^+}=\sum_{i=1}^{d}\ket{ii}$.
It can be checked that
\begin{align}
\tr\left[(A^{A\prec B}_{a|x}\otimes B^{A\prec B}_{b|y})\ W^{A\prec B}\right]  &= p_A^{A\prec B}(a|x)p_B^{A\prec B}(b|axy) \\
&=p^{A\prec B}(ab|xy),
\end{align}
for every $a,b,x,y$. Hence, such construction can recover any behaviour that is causally ordered from Alice to Bob. Analogously, for $\{p^{B\prec A}(ab|xy)\}$, one has
\begin{align}
A^{B\prec A}_{a|x} &= \sum_{b,y}p_A^{B\prec A}(a|bxy)\ketbra{by}{by}^{A_I}\otimes\frac{\1^{A_O}}{d_{A_O}}, \\
B^{B\prec A}_{b|y} &= \1^{B_I}\otimes p_B^{B\prec A}(b|y) \ketbra{by}{by}^{B_O} \\
W^{B\prec A} &=  \frac{\1^{B_I}}{d_{B_I}} \otimes \ketbra{\Phi^+}{\Phi^+}^{B_OA_I} \otimes \1^{A_O}.
\end{align}

Now, for causal behaviours, which are convex combinations of causally ordered behaviours, we construct instruments and process matrices in such a way that each causal order acts on a complementary subspace. That is, we define the valid process matrix
\beq
W^\text{causal} = q W^{A\prec B} \otimes \ketbra{00}{00}^{A^{'}_IB^{'}_I} + (1-q) W^{B\prec A} \otimes \ketbra{11}{11}^{A^{'}_IB^{'}_I},
\eeq
by extending Alice's and Bob's input spaces according to $\set{H}^{A_I(B_I)}\to \set{H}^{A_I(B_I)}\otimes\set{H}^{A^{'}_I(B^{'}_I)} $, and we define the valid instruments
\begin{align}
A_{a|x} &= A^{A\prec B}_{a|x} \otimes \ketbra{0}{0}^{A^{'}_I} + A^{B\prec A}_{a|x} \otimes \ketbra{1}{1}^{A^{'}_I} \\
B_{b|y} &= B^{A\prec B}_{b|y} \otimes \ketbra{0}{0}^{B^{'}_I} + B^{B\prec A}_{b|y} \otimes \ketbra{1}{1}^{B^{'}_I}.
\end{align}

In this way, it is easy to check that orthogonal terms cancel out and we arrive at
\begin{align}
\tr\left[(A_{a|x}\otimes B_{b|y})\ W^\text{causal}\right]  &= qp^{A\prec B}(ab|xy) + (1-q)p^{B\prec A}(ab|xy) \\
&=p^\text{causal}(ab|xy),
\end{align}
recovering any causal behaviour.
\end{proof}

\begin{lemma}\label{lmmass2}
A general assemblage is causal if and only if it is a process assemblage that can be obtained from a causally separable process matrix.
\end{lemma}

\begin{proof}
To prove the if part we show that a causally separable process matrix can only give rise to a causal assemblage regardless of Alice's sets of instruments.
First we show that an assemblage that is generated by a process matrix that is ordered from Alice to Bob is also ordered from Alice to Bob.
\begin{align}
_{B_O}w_{a|x} &= _{B_O} \tr_{A}\left[(A_{a|x}\otimes\1)W^{A\prec B}\right] \\
   					&= \tr_{A}\left[(A_{a|x}\otimes\1)_{B_O}W^{A\prec B}\right] \\
					&= \tr_{A}\left[(A_{a|x}\otimes\1)W^{A\prec B}\right] \\
					&= w_{a|x},
\end{align}
since $W^{A\prec B}=_{B_O}W^{A\prec B}$. Hence, $\{w_{a|x}\}=\{w^{A\prec B}_{a|x}\}$ is causally ordered from Alice to Bob.

Next, we show that an assemblage that is generated by a process matrix that is ordered from Bob to Alice is also ordered from Bob to Alice.
\begin{align}
\sum_a w_{a|x} &= \sum_a \tr_{A}\left[(A_{a|x}\otimes\1)W^{B\prec A}\right] \\
   					&= \tr_{A}\Big[(\sum_aA_{a|x}\otimes\1) _{A_O}W^{B\prec A}\Big] \\
					&= \tr_{A_I}\Big[W^{B_IB_OA_I}(\tr_{A_O}\sum_aA_{a|x}\otimes\1^{B_IB_O})\Big] \\
					&= \tr_{A_I}\left[W^{B_IB_OA_I}(\1^{A_I}\otimes\1^{B_IB_O})\right],
\end{align}
which is independent of $x$. Hence, $\{w_{a|x}\}=\{w^{B\prec A}_{a|x}\}$ is causally ordered from Bob to Alice. 

Consequently, assemblages that come from causally separable process matrices according to
\begin{align}
w_{a|x} &= \tr_{A}\left[(A_{a|x}\otimes\1)W^\text{causal}\right] \\ 
&= \tr_{A}\left[(A_{a|x}\otimes\1) \ (q W^{A\prec B} + (1-q)W^{B\prec A})\right] \\ 
&= q w^{A\prec B}_{a|x} + (1-q) w^{B\prec A}_{a|x}
\end{align}
are causal by definition. In other words, causally separable process matrices can only generate causal assemblages.

To prove the only if part we show that all assemblages that are causal can be reproduced by performing some instruments on some causally separable process matrix.
Given a causal assemblage, it can be decomposed in $\{w^{A\prec B}_{a|x}\}$ and $\{w^{B\prec A}_{a|x}\}$ with some convex weight $q$.

With the contribution that is causally ordered from Alice to Bob, $\{w^{A\prec B}_{a|x}\}$, one can write its elements as $w^{A\prec B}_{a|x}=\sigma_{a|x}^{B_I}\otimes\1^{B_O}$, where $\tr\sum_a\sigma_{a|x}=1$. We show that it can be recovered by instruments
\beq
A^{A\prec B}_{a|x}=\1^{A_I}\otimes\sigma_{a|x}^{T \, A_O},
\eeq
where $^T$ is the transposition in the computational basis of $\set{H}^{A_O}$, and a causally ordered process matrix from Alice to Bob
\beq
W^{A\prec B} = \frac{\1^{A_I}}{d_{A_I}}\otimes\ketbra{\Phi^+}{\Phi^+}^{A_OB_I}\otimes\1^{B_O},
\eeq
where $\ketbra{\Phi^+}{\Phi^+}$ is the Choi operator of the identity channel, with $\ket{\Phi^+}=\sum_{i=1}^{d}\ket{ii}$. Indeed,
\begin{align}
\tr_{A_IA_O}\left[(A^{A\prec B}_{a|x}\otimes\1^{B_IB_O}) \ W^{A\prec B}\right] &= \tr_{A_O}\left[(\sigma_{a|x}^{T \, A_O}\otimes\1^{B_IB_O}) (\ketbra{\Phi^+}{\Phi^+}^{A_OB_I}\otimes\1^{B_O})\right] \\
&= \sigma_{a|x}^{B_I}\otimes\1^{B_O} \\
&= w^{A\prec B}_{a|x}.
\end{align}

On the other hand, with the contribution that is causally ordered from Bob to Alice, $\{w^{B\prec A}_{a|x}\}$, one can write $\sum_a w^{B\prec A}_{a|x}=\rho^{B_I}\otimes\1^{B_O}$, where $\tr\rho^{B_I}=1$. Since $\rho\geq0$, it can be written as $\rho = \sum_i\mu_i\ketbra{i}{i}$, which is purified by $\ket{\psi}=\sum_i\sqrt{\mu_i}\ket{ii}$. We show that $\{w^{B\prec A}_{a|x}\}$ can be recovered by instruments $\{A_{a|x}\}$,
\beq
A^{B\prec A}_{a|x}=(\rho^{-\frac{1}{2}\,A^{'}_{I}}\otimes\1^{A^{''}_{I}})\, w^{T \, A_I}_{a|x}\, (\rho^{-\frac{1}{2}\,A^{'}_{I}}\otimes\1^{A^{''}_{I}})\otimes\frac{\1^{A_O}}{d_{A_O}},
\eeq
where $\rho^{-1}$ is the inverse of $\rho$ on its support,  $\set{H}^{A_I}=\set{H}^{A^{'}_I}\otimes\set{H}^{A^{''}_I}$, $^T$ is the transposition in the $\{\ket{i}\}_i$ basis, and a causally ordered process matrix from Bob to Alice
\beq
W^{B\prec A} = \ketbra{\psi}{\psi}^{B_I A^{'}_I} \otimes \ketbra{\Phi^+}{\Phi^+}^{B_O A^{''}_I}\otimes\1^{A_O}.
\eeq
Indeed,
\begin{align}
\tr_{A_IA_O}\left[(A^{B\prec A}_{a|x}\otimes\1^{B_IB_O}) \ W^{B\prec A}\right] &= w^{B\prec A}_{a|x}.
\end{align}

Finally, just like in the proof of \cref{lmmbehaviour}, by allowing the different causal orders to act in complementary subspaces, we can recover any convex combinations of causally ordered assemblages, \textit{i.e.}, causal assemblages, from causally separable process matrices.
\end{proof}


\section{Characterization theorems for general and causal assemblages}\label{app:assemblages}

In this appendix, we prove the equivalence between the definition induced by \cref{eqmostgenass} and \cref{defbigenass}, that is, we show that the most general set of operators $\{w_{a|x}\}$ that satisfies 
\beq\label{apeqA}
p(ab|xy)=\tr(B_{b|y}w_{a|x}) \ \ \ \forall \ a,b,x,y,
\eeq
where $\{p(ab|xy)\}$ is a general behaviour and $\{B_{b|y}\}$ is a valid set of instruments, is a general assemblage of \cref{defbigenass}.

Next, we prove this equivalence to hold also between the definition induced by \cref{eqmostgencauass} and \cref{defdeforderass}, that is, we show that the most general set of operators $\{w^{A\prec B}_{a|x}\}$ that satisfies
\beq
p^{A\prec B}(ab|xy)=\tr(B_{b|y}w^{A\prec B}_{a|x}) \ \ \ \forall \ a,b,x,y,
\eeq
where $\{p^{A\prec B}(ab|xy)\}$ is a behaviour that is causally ordered from Alice to Bob and $\{B_{b|y}\}$ is a valid set of instruments, is an assemblage that is causally ordered from Alice to Bob, of \cref{defdeforderass}, end equivalently from Bob to Alice.

Let us begin with the conditions of a general assemblage.

Non-negativity yields
\beq\label{apeqA1}
p(ab|xy)\geq0 \ \ \forall\,B_{b|y}\geq0 \iff w_{a|x}\geq 0 \ \ \ \forall \ a,x,
\eeq
while normalization yields
\beq
\sum_{a,b} p(ab|xy)=\sum_{a,b}\tr(B_{b|y}w_{a|x}) =1.
\eeq

Let $B_y=\sum_bB_{b|y}$ be the Choi operator of a CPTP map. Then, following ref.~\cite{araujo15}, we can parametrize it according to
\beq
B_y= _{[1-B_O]}Y +\frac{\1}{d_{B_O}},
\eeq
where $Y$ is an hermitian operator and we define the map $_{[1-B_O]}M = M - _{B_O}M$.

Applying the parametrization, it follows that
\beq
\sum_a\tr\left[\left( _{[1-B_O]}Y +\frac{\1}{d_{B_O}}\right)w_{a|x}\right] =1.
\eeq

If $Y=0$, we have
\beq\label{apeqA2}
\tr\bigg(\frac{1}{d_{B_O}}\sum_a w_{a|x}\bigg) = 1 \iff \tr\sum_a w_{a|x} = d_{B_O}.
\eeq

If $Y\neq0$, and using \cref{apeqA2}, we have
\beq
\tr\bigg( {}_{[1-B_O]}Y \, \sum_a w_{a|x}\bigg) = 0.
\eeq

Due to the self-duality of the map $_{[1-B_O]}\cdot$, 
\beq
\tr \bigg( {}_{[1-B_{O}]}Y \, \sum_a w_{a|x} \bigg) = \tr \bigg( Y \, _{[1-B_{O}]} \sum_a w_{a|x} \bigg) = 0 \ \ \forall \,Y \iff _{[1-B_{O}]} \sum_a w_{a|x} = 0.
\eeq
Hence,
\beq\label{apeqA3}
\sum_a w_{a|x} = _{B_{O}} \sum_a w_{a|x}.
\eeq

Together, \cref{apeqA1,apeqA2,apeqA3} define a general assemblage. It is simple to verify that general assemblages lead to valid probability distributions for all sets of instruments.

For assemblages that are causally ordered from Alice to Bob, we require first that they satisfy the conditions of a general assemblage, in order to lead to valid behaviours, and second that Alice's marginal probability distributions
\beq
\sum_b p^{A\prec B}(ab|xy) = \tr\bigg(\sum_b B_{b|y}w^{A\prec B}_{a|x}\bigg)					  
\eeq
are independent of $y$, so that the resulting behaviour is causally ordered from Alice to Bob. The implication is that
\begin{align}
\sum_b p^{A\prec B}(ab|xy) &=  \tr\left[\left( _{[1-B_O]}Y +\frac{\1}{d_{B_O}}\right) w^{A\prec B}_{a|x}\right] \\
					   &=  \tr\left( _{[1-B_O]}Y \, w^{A\prec B}_{a|x}\right) + \frac{1}{d_{B_O}}\tr\left(w^{A\prec B}_{a|x}\right).
\end{align}
will be independent of $y$ if and only if
\beq
\tr\left( _{[1-B_O]}Y \, w^{A\prec B}_{a|x}\right) = \tr\left( Y \, _{[1-B_O]} w^{A\prec B}_{a|x}\right) = 0 \ \ \forall \ Y \iff _{[1-B_O]} w^{A\prec B}_{a|x} = 0.
\eeq
Hence,
\beq\label{apeqAB}
w^{A\prec B}_{a|x} = _{B_O}w^{A\prec B}_{a|x} \ \ \ \forall \ a,x.
\eeq

Finally, for an assemblage that is causally ordered from Bob to Alice, it is also required that they satisfy the conditions of a general assemblage. To guarantee that Bob's marginal probability distributions
\beq
\sum_a p^{A\prec B}(ab|xy) = \tr\bigg(B_{b|y}\sum_a w^{A\prec B}_{a|x}\bigg)					  
\eeq
are independent of $x$, it is necessary and sufficient that
\beq\label{apeqBA}
\sum_a w^{B\prec A}_{a|x} = \sum_a w^{B\prec A}_{a|x'} \ \ \ \forall \ b,x,x',y.
\eeq

Together, \cref{apeqAB,apeqBA} and the conditions of general assemblages define causally ordered assemblages. It is simple to verify that causal assemblages lead to causal probability distributions for all sets of instruments.

The technique used here to arrive at the definitions of general and causal assemblages is the same that will be used in all tripartite scenarios to define their respective general and causal assemblages.


\section{Causally nonseparable process matrices that cannot be certified \\ in a semi-device-independent way}\label{app:counterexemples}

In this section we present the proof of \cref{thmcounterexample} from the main text. 	

\setcounter{theorem}{5}
\begin{theorem}[Device-dependent certifiable, semi-device-independent noncertifiable process matrix]
There exist causally nonseparable process matrices that, for any sets of instruments on Alice's side, always give rise to causal assemblages. That is, causally nonseparable process matrices that cannot be certified in a semi-device-independent way.

In particular, let $W\in\set{L}(\set{H}^{A_IA_OB_IB_O})$ be a process matrix and $W^{T_A}$ be the partial transposition of $W$ with respect to some basis in $\set{L}(\set{H}^{A_IA_O})$ for Alice. If $W^{T_A}$ is causally separable, the assemblages generated by $w_{a|x}=\tr_A[(A_{a|x}\otimes\1^B)\,W]$ are causal for every set of instruments $\{A_{a|x}\}$.
\end{theorem}

\begin{proof}
We start our proof by showing that if $W^{T_A}$ is causally separable then the assemblages generated by $w_{a|x}=\tr_A[(A_{a|x}\otimes\1^B)\,W]$ are causal for every set of instruments $\{A_{a|x}\}$. Straightforward calculations shows that the transposition map is self-adjoint (\textit{i.e.} $\tr[A^T\, B]= \tr[A\, B^T], \; \forall \ A, B$), in particular, it is true that
\begin{equation}
\tr_A \left[ (A_{a|x}\otimes \id )\, W^{T_A}  \right] =  \tr_A{ \left[ ( A_{a|x}^T \otimes \id )\,W  \right] } \quad \forall \ A_{a|x}.
\end{equation}

Now notice that every instrument can be written as a transposition of another instrument and we can define new valid instruments via $A'_{a|x}:=A^T_{a|x}$, which allows us to recover any assemblages generated by $W$ and instruments $\{A_{a|x}\}$ with a causally separable process matrix. More precisely, since $(A_{a|x}^T)^T=A_{a|x}$, the mathematical identities
\begin{align}
		w_{a|x} &=\tr_A[(A_{a|x}\otimes\1)\,W] \\
		& = \tr_A\big[((A_{a|x}^T)^T\otimes\1)\, W\big] \\
		& = \tr_A\big[(A^T_{a|x}\otimes\1)\, W^{T_A}\big] \\
		& = \tr_A[(A'_{a|x}\otimes\1)\, W^{\text{sep}}] 
\end{align}
provide an explicit decomposition for the assemblage $\{w_{a|x}\}$ in terms of a causally separable process matrix and valid instruments.
	
We finish our proof by referring to ref.~\cite{feix16}, which presents several examples of process matrices $W$ which are causally nonseparable but that $W^{T_A}$ is causally separable.
\end{proof}


\section{Certification of tripartite process matrices}\label{app:tripartite}

\setcounter{theorem}{7}
\setcounter{lemma}{3}

Here we generalize our approach to certification of indefinite causal order for a tripartite scenario where a process matrix is shared between parties Alice, Bob, and Charlie. In this particular tripartite scenario, Charlie is always in the future of Alice and Bob, so we only study tripartite process matrices that have the causal order $(A,B)\prec C$, however, the causal order between Alice and Bob may or may not be well defined. We will only consider this kind of tripartite scenario, the one which is appropriate for the study of the quantum switch processes. {A semi-device-independent approach to more general tripartite scenarios, or multipartite scenarios, may be derived, in principle, from a straightforward generalization of our particular tripartite scenario.} For more general multipartite scenarios under a device-dependent approach, we refer the reader to ref.~\cite{wechs19}, and for device-independent, ref.~\cite{abbott17}.

In the following, we explicitly extend all concepts, definitions, and results from the bipartite case presented in the main text to the tripartite case. We start by defining certification in all $6$ inequivalent scenarios that arise from making different assumptions about the operations of each party: TTT (device-dependent), UUU (device-independent), TTU, TUU, UTT, and UUT (semi-device-independent).

\begin{definition}[TTT (device-dependent) certification]
Given a tripartite behaviour $\{p(abc|xyz)\}$ that arises from known instruments $\{\overline{A}_{a|x}\}$ and $\{\overline{B}_{b|y}\}$, known POVMs $\{\overline{M}_{c|z}\}$, and an unknown tripartite process matrix, one certifies that this process matrix is causally nonseparable if, for some $a,b,c,x,y,z$,
\beq
p(abc|\overline{A}_{a|x}\,,\overline{B}_{b|y},\overline{M}_{c|z})\neq\tr\left[(\overline{A}_{a|x}\otimes \overline{B}_{b|y}\otimes \overline{M}_{c|z})W^\text{sep}\right],
\eeq
for all causally separable tripartite process matrices $W^\text{sep}$.
\end{definition}

\begin{definition}[UUU (device-independent) certification]
Given a tripartite behaviour $\{p(abc|xyz)\}$ that arises from unknown instruments and an unknown tripartite process matrix, one certifies that this process matrix is causally nonseparable if, for some $a,b,c,x,y,z$,
\beq
p(abc|xyz)\neq\tr\left[(A_{a|x}\otimes B_{b|y}\otimes M_{c|z})W^\text{sep}\right],
\eeq
for all causally separable tripartite process matrices $W^\text{sep}$, all general instruments $\{A_{a|x}\}$ and $\{B_{b|y}\}$, and all general POVMs $\{M_{c|z}\}$. A process matrix certified in such way is called UUU-noncausal, or device-independent noncausal.
\end{definition}

\begin{definition}[TTU (semi-device-dependent) certification]
Given a tripartite behaviour $\{p(abc|xyz)\}$ that arises from known instruments $\{\overline{A}_{a|x}\}$ and $\{\overline{B}_{b|y}\}$ on Alice's and Bob's side, unknown POVMs on Charlie's side, and an unknown tripartite process matrix, one certifies that this process matrix is causally nonseparable if, for some $a,b,c,x,y,z$,
\beq
p(abc|\overline{A}_{a|x},\overline{B}_{b|y},\,z)\neq\tr\left[(\overline{A}_{a|x}\otimes \overline{B}_{b|y}\otimes M_{c|z})W^\text{sep}\right],
\eeq
for all causally separable tripartite process matrices $W^\text{sep}$ and all general POVMs $\{M_{c|z}\}$. A process matrix certified in such way is called TTU-noncausal.
\end{definition}

\begin{definition}[TUU (semi-device-dependent) certification]
Given a tripartite behaviour $\{p(abc|xyz)\}$ that arises from known instruments $\{\overline{A}_{a|x}\}$ on Alice's, unknown instruments on Bob's and Charlie's side, and an unknown tripartite process matrix, one certifies that this process matrix is causally nonseparable if, for some $a,b,c,x,y,z$,
\beq
p(abc|\overline{A}_{a|x},\,y,\,z)\neq\tr\left[(\overline{A}_{a|x}\otimes B_{b|y}\otimes M_{c|z})W^\text{sep}\right],
\eeq
for all causally separable tripartite process matrices $W^\text{sep}$, all general instruments $\{B_{b|y}\}$, and all general POVMs $\{M_{c|z}\}$. A process matrix certified in such way is called TUU-noncausal.
\end{definition}

\begin{definition}[UTT (semi-device-dependent) certification]
Given a tripartite behaviour $\{p(abc|xyz)\}$ that arises from unknown instruments on Alice's side, known instruments $\{\overline{B}_{b|y}\}$ and $\{\overline{M}_{c|z}\}$ on Bob's and Charlie's side, and an unknown tripartite process matrix, one certifies that this process matrix is causally nonseparable if, for some $a,b,c,x,y,z$,
\beq
p(abc|x\,,\overline{B}_{b|y},\overline{M}_{c|z})\neq\tr\left[(A_{a|x}\otimes \overline{B}_{b|y}\otimes \overline{M}_{c|z})W^\text{sep}\right],
\eeq
for all causally separable tripartite process matrices $W^\text{sep}$ and all general instruments $\{A_{a|x}\}$. A process matrix certified in such way is called UTT-noncausal.
\end{definition}

\begin{definition}[UUT (semi-device-dependent) certification]
Given a tripartite behaviour $\{p(abc|xyz)\}$ that arises from unknown instruments on Alice's and Bob's side, known POVMs $\{\overline{M}_{c|z}\}$ on Charlie's side, and an unknown tripartite process matrix, one certifies that this process matrix is causally nonseparable if, for some $a,b,c,x,y,z$,
\beq
p(abc|x,\,y,\,\overline{M}_{c|z})\neq\tr\left[(A_{a|x}\otimes B_{b|y}\otimes \overline{M}_{c|z})W^\text{sep}\right],
\eeq
for all causally separable tripartite process matrices $W^\text{sep}$ and all general instruments $\{A_{a|x}\}$ and $\{B_{b|y}\}$. A process matrix certified in such way is called UUT-noncausal.
\end{definition}


\subsection{Device-dependent -- TTT}

Following ref.~\cite{araujo15}, a \textit{tripartite process matrix}, with Charlie in the future of Alice and Bob, is an operator $W\in\set{L}(\set{H}^{A_I}\otimes\set{H}^{A_O}\otimes\set{H}^{B_I}\otimes\set{H}^{B_O}\otimes\set{H}^{C_I})$ that satisfies
\begin{align}
W&\geq0 \\
\tr\ W &= d_{A_O}d_{B_O}  \\ 
_{A_IA_OC_I}W&=_{A_IA_OB_OC_I}W  \\
_{B_IB_OC_I}W&=_{A_OB_IB_OC_I}W  \\
_{C_I}W &= _{A_OC_I}W+_{B_OC_I}W-_{A_OB_OC_I}W.
\end{align}

A tripartite process matrix $W^{A\prec B\prec C}$ is causally ordered from Alice to Bob to Charlie if it satisfies
\begin{align}
_{C_I} W^{A\prec B\prec C} &= _{B_OC_I} W^{A\prec B\prec C} \label{eqWABC} \\
_{B_IB_OC_I} W^{A\prec B\prec C} &= _{A_OB_IB_OC_I} W^{A\prec B\prec C}, 
\end{align}
and a tripartite process matrix $W^{B\prec A\prec C}$ is causally ordered from Bob to Alice to Charlie if it satisfies
\begin{align}
_{C_I} W^{B\prec A\prec C} &= _{A_OC_I} W^{B\prec A\prec C} \\
_{A_IA_OC_I} W^{B\prec A\prec C} &= _{A_IA_OB_OC_I} W^{B\prec A\prec C}.
\end{align}

A tripartite process matrix $W^\text{sep}$ is causally separable if it can be expressed as a convex combination of causally ordered process matrices of the kind $W^{A\prec B\prec C}$ and $W^{B\prec A\prec C}$, \textit{i.e.},
\beq\label{eqTTTcausal}
W^\text{sep} \coloneqq q W^{A\prec B\prec C} + (1-q) W^{B\prec A\prec C}
\eeq
where $0\leq q\leq 1$ is a real number. A tripartite process matrix that does not satisfy \cref{eqTTTcausal} is called causally nonseparable.

Just as in the bipartite case, all causally nonseparable tripartite process matrices can be certified in a device-dependent way for some choice of instruments, since tomographically complete instruments allow for the full characterization of the process matrix. Hence, \cref{thmdd} also holds in the tripartite case.


\subsection{Device-independent -- UUU}

A \textit{tripartite behaviour} $\{p(abc|xyz)\}$ is a set of joint probability distributions, that is, a set in which each element $p(abc|xyz)$ is a real number, such that
\begin{align}
p(abc|xyz) &\geq 0 \ \ \ \forall \ a,b,c,x,y,z \\
\sum_{a,b,c} p(abc|xyz) &= 1 \ \ \ \forall \ x,y,z,
\end{align}
where $a\in\{1,\ldots,O_A\}$, $b\in\{1\ldots,O_B\}$, and $c\in\{1\ldots,O_C\}$ label the outcomes and $x\in\{1,\ldots,I_A\}$, $y\in\{1,\ldots,I_B\}$, and $z\in\{1,\ldots,I_C\}$ label the inputs. 

A tripartite behaviour in which Charlie is in the future of Alice and Bob is a behaviour $\{p(abc|xyz)\}$ that satisfies
\beq
\sum_c p(abc|xyz) = \sum_c p(abc|xyz') \ \ \ \forall \ a,b,x,y,z,z',
\eeq
that is, a behaviour whose Alice and Bob's joint marginal does not depend on Charlie's inputs. We will only consider this type of behaviours and will refer to them as simply tripartite behaviours.

A tripartite behaviour is called a \textit{tripartite process behaviour} if there exist a tripartite process matrix $W^{A_IA_OB_IB_OC_I}$, sets of instruments $\{A^{A_IA_O}_{a|x}\}$ and $\{B^{B_IB_O}_{b|y}\}$, and a set of POVMs $\{M^{C_I}_{c|z}\}$ such that
\beq\label{eqgborntri}
p^Q(abc|xyz) = \tr\left[(A^{A_IA_O}_{a|x}\otimes B^{B_IB_O}_{b|y}\otimes M^{C_I}_{c|z})\ W^{A_IA_OB_IB_OC_I}\right]  \ \ \ \forall \ a,b,c,x,y,z.
\eeq

Just like in the bipartite case, it can be shown that the tripartite process matrix is the most general operator that leads to valid tripartite behaviours when taken the trace with product instruments. Additionally, since not all bipartite behaviours can be realized by process matrices and the bipartite case is a particular case of the tripartite case, not all tripartite behaviours can be realized by process matrices as well. Hence, \cref{thmprobtwoway} also holds in the tripartite case.

A tripartite behaviour $\{p^{A\prec B\prec C}(abc|xyz)\}$ is causally ordered from Alice to Bob to Charlie if Alice's marginals do not depend on the inputs of Bob and Charlie, that is,
\beq
\sum_{b,c} p^{A\prec B\prec C}(abc|xyz) = \sum_{b,c} p^{A\prec B\prec C}(abc|xy'z') \ \ \ \forall \ a,x,y,y',z,z', 
\eeq 
and a tripartite behaviour $\{p^{B\prec A\prec C}(abc|xyz)\}$ is causally ordered from Bob to Alice to Charlie if Bob's marginals do not depend on the inputs of Alice and Charlie, that is,
\beq
\sum_{a,c} p^{B\prec A\prec C}(abc|xyz) = \sum_{a,c} p^{B\prec A\prec C}(abc|x'yz') \ \ \ \forall \ a,x,x',y,z,z'.
\eeq 

A tripartite behaviour $\{p^\text{causal}(abc|xyz)\}$ is causal if it can be written as a convex combination of causally ordered behaviours of the kind $\{p^{A\prec B\prec C}(abc|xyz)\}$ and $\{p^{B\prec A\prec C}(abc|xyz)\}$, \textit{i.e.},
\beq\label{eqUUUcausal}
p^\text{causal}(abc|xyz) \coloneqq q p^{A\prec B\prec C}(abc|xyz) + (1-q) p^{B\prec A\prec C}(abc|xyz),
\eeq
where $0\leq q\leq 1$ is a real number. A tripartite behaviour that does not satisfy \cref{eqUUUcausal} is called noncausal.

\begin{lemma}\label{lmmcausaltribehaviour}
A tripartite behaviour is causal if and only if it is a tripartite process behaviour that can be obtained by a causally separable tripartite process matrix.
\end{lemma}

\begin{proof}
It can be straightforwardly checked that a behaviour that comes from a causally separable process matrix is causal, the tripartite case being analogous to the bipartite case (see \cref{lmmbehaviour}). 
To prove that all causal tripartite behaviours can be reproduced by a causally separable tripartite process matrix, we provide the explicit construction of instruments and process matrix below.
Given a causal tripartite behaviour $\{p^\text{causal}(abc|xyz)\}$, one can write its decompostion into definite causal orders using $\{p^{A\prec B\prec C}(abc|xyz)\}$, $\{p^{B\prec A\prec C}(abc|xyz)\}$, and $q$. According to Bayes' rule and the nonsignaling principle we can calculate the quantities
\begin{align}
p^{A\prec B\prec C}(abc|xyz) &= p^{A\prec B\prec C}_A(a|xyz)p^{A\prec B\prec C}_B(b|axyz)p^{A\prec B\prec C}_C(c|abxyz) \\
&= p^{A\prec B\prec C}_A(a|x)p^{A\prec B\prec C}_B(b|axy)p^{A\prec B\prec C}_C(c|abxyz).
\end{align} 
We use them to define instruments
\begin{align}
A^{A\prec B\prec C}_{a|x} &= \1^{A_I} \otimes p^{A\prec B\prec C}_A(a|x) \ketbra{ax}{ax}^{A_O} \\
B^{A\prec B\prec C}_{b|y} &= \sum_{a,x} p^{A\prec B\prec C}_B(b|axy) \ketbra{ax}{ax}^{B_I} \otimes  \ketbra{abxy}{abxy}^{B_O} \\
M^{A\prec B\prec C}_{c|z} &= \sum_{a,b,x,y} p^{A\prec B\prec C}_C(c|abxyz) \ketbra{abxy}{abxy}^{C_I},
\end{align} 
and the process matrix
\beq
W^{A\prec B\prec C} = \frac{\1^{A_I}}{d_{A_I}}\otimes \ketbra{\Phi^+}{\Phi^+}^{A_OB_I}\otimes \ketbra{\Phi^+}{\Phi^+}^{B_OC_I}.
\eeq
One can check that
\begin{align}
\tr&\left[(A^{A\prec B\prec C}_{a|x}\otimes B^{A\prec B\prec C}_{b|y}\otimes M^{A\prec B\prec C}_{c|z})\ W^{A\prec B\prec C} \right] = \\
&= p^{A\prec B\prec C}_A(a|x)p^{A\prec B\prec C}_B(b|axy)p^{A\prec B\prec C}_C(c|abxyz) \\
&=p^{A\prec B\prec C}(abc|xyz).
\end{align} 

Equivalently, the instruments and process matrix for the causal order $B\prec A\prec C$ can be constructed, and by allowing each order to act on a complementary input subspace, just like in the proof of \cref{lmmbehaviour}, all causal tripartite behaviours can be recovered. 
\end{proof}

\begin{theorem}
A tripartite process matrix is certified to be causally nonseparable in a device-independent way if and only if it can generate a noncausal tripartite behaviour for some choice of instruments for Alice and Bob and some choice of POVMs for Charlie.
\end{theorem}
The proof is analogous to \cref{thmdi}.


\subsection{Semi-device-independent -- TTU}

We start the first semi-device-independent tripartite scenario by defining assemblages in the TTU scenario using the same reasoning and motivation as the bipartite case, explained in the main text.

\begin{definition}[Process TTU-assemblage]
A process TTU-assemblage is a set of operators $\{w^Q_{c|z}\}$, $w^Q_{c|z}\in\set{L}(\set{H}^{A_IA_O}\otimes\set{H}^{B_IB_O})$, for which there exists a tripartite process matrix $W^{A_IA_OB_IB_OC_I}$ and a set of POVMs $\{M^{C_I}_{c|z}\}$ such that
\beq\label{processwttu}
w^Q_{c|z} = \tr_{C_I}\left[(\1^{A_IA_O}\otimes\1^{B_IB_O}\otimes M^{C_I}_{c|z})W^{A_IA_OB_IB_OC_I}\right],
\eeq
for all $c,z$.
\end{definition}

\begin{definition}[General TTU-assemblage]\label{defttuass}
A general TTU-assemblage is a set of operators $\{w_{c|z}\}$, $w_{c|z}\in\set{L}(\set{H}^{A_IA_O}\otimes\set{H}^{B_IB_O})$, that satisfies
\begin{align}
w_{c|z} &\geq 0 \ \ \ \forall \ c,z \\
\sum_c w_{c|z} &= W^{A_IA_OB_IB_O} \ \ \ \forall \ z,
\end{align}
where $W^{A_IA_OB_IB_O}$ is a valid bipartite process matrix.
\end{definition}

Intuitively, behaviours are extracted from TTU-assemblages according to
\beq
p(abc|xyz) = \tr\left[(A_{a|x}\otimes B_{b|y}) w_{c|z}\right] \ \ \ \forall \ a,b,c,x,y,z.
\eeq

Contrarily to the bipartite case, for which we proved that not all general assemblages can be realized by process matrices (\cref{thmgenass}), for the TTU scenario, we prove that general and process TTU-assemblages are actually equivalent.

\begin{theorem}\label{thmschroedingerttu}
A general TTU-assemblage is valid if and only if it is a process TTU-assemblage. 
\end{theorem}

\begin{proof}
By substituting the definition of a tripartite process matrix and POVMs into \cref{processwttu} it is easy to check that the resulting assemblage satisfies \cref{defttuass}.
To show that any valid TTU-assemblage can be obtained with process matrices and instruments, we give the following explicit construction.
Given a general TTU-assemblage $\{w_{c|z}\}$ and the general bipartite process matrix $W=\sum_c w_{c|z}$, we construct Charlie's POVMs $\{M_{c|z}\}$ according to
\beq
M_{c|z} = W^{-\frac{1}{2}} w^T_{c|z} W^{-\frac{1}{2}}, \ \ \ \forall \ c,z,
\eeq
where the transpose $^T$ is taken in the basis in which $W$ is diagonal.
Notice that this implies $\text{dim}(\set{H}^{C_I})=\text{dim}(\set{H}^{A_IA_O}\otimes\set{H}^{B_IB_O})$ for this particular construction.
The sum $\sum_c M_{c|z} = W^{-\frac{1}{2}} \sum_c w^T_{c|z} W^{-\frac{1}{2}} = W^{-\frac{1}{2}} W W^{-\frac{1}{2}} = \1$ for all $z$ guarantees it is a valid set of POVMs. 

Now, by writing the process matrix $W$ in its diagonal basis, $W=\sum_{ij} \mu_{ij} \ketbra{ij}{ij}$, we can define its purification 
\beq
\ket{W^{ABC}}=\sum_{ij}\sqrt{\mu_{ij}}\ket{ij\,ij}.
\eeq
The object $W^{ABC}=\ketbra{W^{ABC}}{W^{ABC}}$ is a well defined tripartite process matrix. This is true particularly because the dimension of Charlie's output space is $1$, that is, because it is in the future of Alice and Bob, and follows from the fact that $\tr_C W^{ABC} = W$. Hence, 
\begin{align}
\tr_{C}\left[(\1^{A}\otimes\1^{B}\otimes M^{C}_{c|z})W^{ABC}\right] &= \tr_{C}\left[(\1^{AB}\otimes W^{-\frac{1}{2}} w^T_{c|z} W^{-\frac{1}{2}})W^{ABC}\right] \\
&= w_{c|z},
\end{align}
for all $c,z$. This concludes the proof that all TTU-assemblages can be realized with valid tripartite process matrices and a set of POVMs for Charlie.
\end{proof}

We now define our notion of causality for TTU-assemblages.

\begin{definition}[Causal TTU-assemblage]
A TTU-assemblage is causally ordered from Alice to Bob to Charlie if it satisfies
\beq
\sum_c w^{A\prec B \prec C}_{c|z} = W^{A\prec B} \ \ \ \forall \ z,
\eeq
where $W^{A\prec B}$ is a bipartite process matrix causally ordered from Alice to Bob, and equivalently from Bob to Alice.

A TTU-assemblage $\{w^{\text{causal}}_{c|z}\}$ is causal if it can be expressed as a convex combination of causally ordered TTU-assemblages, \textit{i.e.},
\beq\label{eqTTUcausal}
w^{\text{causal}}_{c|z} \coloneqq q w^{A\prec B\prec C}_{c|z} + (1-q) w^{B\prec A\prec C}_{c|z} \ \ \ \forall \ c,z,
\eeq
where $0\leq q\leq 1$ is a real number. A TTU-assemblage that does not satisfy \cref{eqTTUcausal} is called a noncausal TTU-assemblage.
\end{definition}

Notice that one can decide whether a TTU-assemblage is causal by means of SDP.

\begin{lemma}\label{lmmttuass}
A TTU-assemblage is causal if and only if it is a process TTU-assemblage that can be obtained from a causal tripartite process matrix.
\end{lemma}
The proof is analogous to the one in \cref{thmschroedingerttu}. Notice that in this case $W^{ABC}$ will be the purification of a causally ordered process matrix, and therefore, also causally ordered.

\begin{theorem}
A tripartite process matrix is certified to be causally nonseparable in a semi-device-independent TTU way if and only if it can generate a noncausal TTU-assemblage for some choice of POVMs for Charlie.
\end{theorem}
The proof is analogous to \cref{thmsdi}.


\subsection{Semi-device-independent -- TUU}

The other semi-device-independent cases follow analogously. We continue with the TUU scenario.

\begin{definition}[Process TUU-assemblage]
A process TUU-assemblage is a set of operators $\{w_{bc|yz}\}$, $w_{bc|yz}\in\set{L}(\set{H}^{A_IA_O})$, for which there exist a tripartite process matrix $W^{A_IA_OB_IB_OC_I}$, a set of instruments $\{B^{B_IB_O}_{b|y}\}$, and a set of POVMs $\{M^{C_I}_{c|z}\}$ such that
\beq\label{eqassprocessTUU}
w^Q_{bc|yz} = \tr_{B_IB_OC_I}\left[(\1^{A_IA_O}\otimes B^{B_IB_O}_{b|y}\otimes M^{C_I}_{c|z})W^{A_IA_OB_IB_OC_I}\right], \ \ \ \forall \ b,c,y,z.
\eeq
\end{definition}

\begin{definition}[General TUU-assemblage]
A general TUU-assemblage is a set of operators $\{w_{bc|yz}\}$, $w_{bc|yz}\in\set{L}(\set{H}^{A_IA_O})$, that satisfies
\begin{align}
w_{bc|yz} &\geq 0 \ \ \ \forall \ b,c,y,z \\
\tr \sum_{b,c} w_{bc|yz} &= d_{A_O} \ \ \ \forall \ y,z \\
\sum_c w_{bc|yz} &= \sum_c w_{bc|yz'} \ \ \ \forall \ b,y,z,z' \\
\sum_{b,c} w_{bc|yz} &= _{A_O} \sum_{b,c} w_{bc|yz} \ \ \ \forall \ y,z.
\end{align}
\end{definition}

Intuitively, behaviours are extracted from TUU-assemblages according to
\beq
p(abc|xyz) = \tr\left(A_{a|x} \ w_{bc|yz}\right) \ \ \ \forall \ a,b,c,x,y,z.
\eeq

In the bipartite case, we have proven that not all general bipartite assemblages can be realized by process matrices (are process bipartite assemblages), while in the previous tripartite case, TTU, we have proven that all general TTU-assemblages can indeed be realized by process matrices (are process TTU-assemblages). However, in this tripartite scenario, TUU, as well as in the remaining cases, UTT and UUT, it is not clear whether all general TUU-, UTT-, and UUT-assemblages can be realized by process matrices. We leave this problem as an open question.

Nevertheless, all process TUU-assemblages are valid general TUU-assemblages.

We now define our notion of causality for TUU-assemblages.

\begin{definition}[Causal TUU-assemblage]\label{defcausalTUUass}
A TUU-assemblage is causally ordered from Alice to Bob to Charlie if it satisfies
\beq
\sum_{b,c} w^{A\prec B\prec C}_{bc|yz} = \sum_{b,c} w^{A\prec B\prec C}_{bc|y'z'}  \ \ \ \forall \ y,y',z,z',
\eeq
and from Bob to Alice to Charlie if it satisfies
\beq
\sum_c w^{B\prec A\prec C}_{bc|yz} = _{A_O} \sum_c w^{B\prec A\prec C}_{bc|yz}  \ \ \ \forall \ b,y,z.
\eeq

A TUU-assemblage $\{w^{\text{causal}}_{bc|yz}\}$ is causal if it can be expressed as a convex combination of causally ordered TUU-assemblages, \textit{i.e.},
\beq\label{eqTUUcausal}
w^{\text{causal}}_{bc|yz} \coloneqq q w^{A\prec B\prec C}_{bc|yz} + (1-q) w^{B\prec A\prec C}_{bc|yz} \ \ \ \forall \ b,c,y,z,
\eeq
where $0\leq q\leq 1$ is a real number. A TUU-assemblage that does not satisfy \cref{eqTUUcausal} is called a noncausal TUU-assemblage.
\end{definition}

Notice that one can decide whether a TUU-assemblage is causal by means of an SDP.

All causally separable tripartite process matrices lead to causal TUU-assemblages, for whatever choice of instruments. Whether all causal TUU-assemblages can be written in terms of causally separable process matrices is not clear, although for the particular case of assemblages that are causally ordered from Alice to Bob to Charlie, we can show this to be the case. 

Here we present our explicit construction of any TUU-assemblage that is causally ordered from Alice to Bob to Charlie by a tripartite process matrix that is causally ordered from Alice to Bob to Charlie.

Given a TUU-assemblage $\{w^{A\prec B\prec C}_{bc|yz}\}$, we can define the state $\rho$ such that $\sum_{b,c} w^{A\prec B\prec C}_{bc|yz} = \rho^{A_I} \otimes \1^{A_O}$. Since $\rho\geq0$ it can be written as $\rho=\sum_i \mu_i \ketbra{i}{i}$, which can be purified by $\ket{psi}=\sum_{i}\sqrt{\mu_i}\ket{ii}$. Then, let $\{B^{A\prec B\prec C}_{b|y}\}$ and $\{M^{A\prec B\prec C}_{c|z}\}$ be instruments
\begin{align}
B^{A\prec B\prec C}_{b|y} &=\1^{B_I} \otimes \ketbra{by}{by}^{B_O} \\
M^{A\prec B\prec C}_{c|z} &= \sum_{b,y} \left(\rho^{-\frac{1}{2}\, C^{'}_I}\otimes\1^{C^{''}_I} \right) w^{T \ A\prec B\prec C \ C^{'}_IC^{''}_I}_{bc|yz} \left(\rho^{-\frac{1}{2}\, C^{'}_I}\otimes\1^{C^{''}_I} \right) \otimes \ketbra{by}{by}^{C^{'''}_I},
\end{align}
where $\rho^{-1}$ be the inverse of $\rho$ on its support, the transpose $^T$ is taken on the $\{\ket{i}\}_i$ basis  and $\set{H}^{C_I} = \set{H}^{C^{'}_I}\otimes\set{H}^{C^{''}_I}\otimes\set{H}^{C^{''}_I}$. Let 
\beq
W^{A\prec B\prec C} = \ketbra{\psi}{\psi}^{C^{'}_IA_I}\otimes\ketbra{\Phi^+}{\Phi^+}^{C^{''}_IA_O}\otimes\frac{\1^{B_I}}{d_{B_I}}\otimes\ketbra{\Phi^+}{\Phi^+}^{C^{'''}_IB_O}
\eeq
be a tripartite process matrix that is causally ordered from Alice to Bob to Charlie. Then, it is true that the assemblage $\{w^{A\prec B}_{bc|yz}\}$ can be recovered by
\beq
\tr_{B_IB_OC_I}\left[(\1^{A}\otimes B^{A\prec B\prec C}_{b|y}\otimes M^{A\prec B\prec C}_{c|z})W^{A\prec B\prec C}\right] = w^{A\prec B\prec C}_{bc|yz}.
\eeq


\subsection{Semi-device-independent -- UTT}

Here we detail our concepts and definitions for the UTT scenario.

\begin{definition}[Process UTT-assemblage]
A process UTT-assemblage is a set of operators $\{w_{a|x}\}$, $w_{a|x}\in\set{L}(\set{H}^{B_IB_O}\otimes\set{H}^{C_I})$, for which there exist a tripartite process matrix $W^{A_IA_OB_IB_OC_I}$ and a set of instruments $\{A^{A_IA_O}_{a|x}\}$ such that
\beq\label{eqprocessassutt}
w^Q_{a|x} = \tr_{A_IA_O}\left[(A^{A_IA_O}_{a|x}\otimes \1^{B_IB_O}\otimes \1^{C_I})W^{A_IA_OB_IB_OC_I}\right], \ \ \ \forall \ a,x.
\eeq
\end{definition}

\begin{definition}[General UTT-assemblage]
A general UTT-assemblage is a set of operators $\{w_{a|x}\}$, $w_{a|x}\in\set{L}(\set{H}^{B_IB_O}\otimes\set{H}^{C_I})$, that satisfies
\begin{align}
w_{a|x} &\geq 0 \ \ \ \forall \ a,x \\
\tr \sum_a w_{a|x} &= d_{B_O} \ \ \ \forall \ x \\
_{C_I}\sum_a w_{a|x} &= _{B_OC_I} \sum_a w_{a|x} \ \ \ \forall \ x.
\end{align}
\end{definition}

Intuitively, behaviours are extracted from UTT-assemblages according to
\beq
p(abc|xyz) = \tr\left[w_{a|x} (B_{b|y}\otimes M_{c|z})\right] \ \ \ \forall \ a,b,c,x,y,z.
\eeq

\begin{definition}[Causal UTT-assemblage]\label{defcausalUTTass}
A UTT-assemblage is causally ordered from Alice to Bob to Charlie if it satisfies
\beq
_{C_I}w^{A\prec B\prec C}_{a|x} = _{B_OC_I}w^{A\prec B\prec C}_{a|x} \ \ \ \forall \ a,x,
\eeq
and from Bob to Alice to Charlie if it satisfies
\beq
_{C_I}\sum_a w^{B\prec A\prec C}_{a|x} = _{C_I}\sum_a w^{B\prec A\prec C}_{a|x'} \ \ \ \forall \ x,x'.
\eeq

A UTT-assemblage $\{w^{\text{causal}}_{a|x}\}$ is causal if it can be expressed as a convex combination of causally ordered UTT-assemblages, \textit{i.e.},
\beq\label{eqUTTcausal}
w^{\text{causal}}_{a|x} \coloneqq q w^{A\prec B\prec C}_{a|x} + (1-q) w^{B\prec A\prec C}_{a|x} \ \ \ \forall \ a,x,
\eeq
where $0\leq q\leq 1$ is a real number. A UTT-assemblage that does not satisfy \cref{eqUTTcausal} is called a noncausal UTT-assemblage.
\end{definition}

Notice that one can decide whether a UTT-assemblage is causal by means of an SDP.

It is not clear whether all general UTT-assemblages can be realized by tripartite process matrices nor whether all causal UTT-assemblages can be realized by causally separable tripartite process matrices. 
Nevertheless, the set of general UTT-assemblages is an outer approximation of the set of process UUT-assemblages and the set of causal UUT-assemblages is an outer approximation of the set of process UTT-assemblages that come from causally separable process matrices. What is not clear is whether these approximations are tight.


\subsection{Semi-device-independent -- UUT}

The final tripartite semi-device independent scenario studied in this work is the UUT scenario.

\begin{definition}[Process UUT-assemblage]
A process UUT-assemblage is a set of operators $\{w_{ab|xy}\}$, $w_{ab|xy}\in\set{L}(\set{H}^{C_I})$, for which there exist a tripartite process matrix $W^{A_IA_OB_IB_OC_I}$ and sets of instruments $\{A^{A_IA_O}_{a|x}\}$, $\{B^{B_IB_O}_{b|y}\}$ such that
\beq
w^Q_{ab|xy} = \tr_{A_IA_OB_IB_O}\left[(A^{A_IA_O}_{a|x}\otimes B^{B_IB_O}_{b|y}\otimes \1^{C_I})W^{A_IA_OB_IB_OC_I}\right], \ \ \ \forall \ a,b,x,y.
\eeq
\end{definition}

\begin{definition}[General UUT-assemblage]
A general UUT-assemblage is a set of operators $\{w_{ab|xy}\}$, $w_{ab|xy}\in\set{L}(\set{H}^{C_I})$, that satisfies
\begin{align}
w_{ab|xy} &\geq 0 \ \ \ \forall \ a,b,x,y \\
\tr \sum_{a,b} w_{ab|xy} &= 1 \ \ \ \forall \ x,y.
\end{align}
\end{definition}

Intuitively, behaviours are extracted from UUT-assemblages according to
\beq
p(abc|xyz) = \tr\left(w_{ab|xy} \ M_{c|z}\right) \ \ \ \forall \ a,b,c,x,y,z.
\eeq

The set of general UUT-assemblages is an outer approximation of the set of process UUT-assemblages but it is not clear to us whether this approximation is tight, \textit{i.e.}, it is not clear whether or not all general UUT-assemblages can be obtained by tripartite process matrices. 

\begin{definition}[Causal UUT-assemblage]
A UUT-assemblage is causally ordered from Alice to Bob to Charlie if it satisfies
\beq
\tr \sum_b w^{A\prec B\prec C}_{ab|xy} = \tr \sum_b w^{A\prec B\prec C}_{ab|xy'} \ \ \ \forall \ a,x,y,y',
\eeq
and from Bob to Alice to Charlie if it satisfies
\beq
\tr \sum_a w^{B\prec A\prec C}_{ab|xy} = \tr \sum_a w^{B\prec A\prec C}_{ab|x'y}  \ \ \ \forall \ b,x,x',y.
\eeq

A UUT-assemblage $\{w^{\text{causal}}_{ab|xy}\}$ is causal if it can be expressed as a convex combination of causally ordered UUT-assemblages, \textit{i.e.},
\beq\label{eqUUTcausal}
w^{\text{causal}}_{ab|xy} \coloneqq q w^{A\prec B\prec C}_{ab|xy} + (1-q) w^{B\prec A\prec C}_{ab|xy} \ \ \ \forall \ a,b,x,y,
\eeq
where $0\leq q\leq 1$ is a real number. A UUT-assemblage that does not satisfy \cref{eqUUTcausal} is called a noncausal UUT-assemblage.
\end{definition}

Notice that one can decide whether a UUT-assemblage is causal by means of an SDP. 

For the case of causal UUT-assemblages, we prove that our approximation is indeed tight, \textit{i.e.}, that all causal UUT-assemblages can be realized by causal tripartite process matrix, analogously to \cref{lmmass} in the bipartite case and \cref{lmmttuass} in the TTU tripartite case.

\begin{lemma}
A UUT-assemblage is causal if and only if it is a process UUT-assemblage that can be obtained from a causal tripartite process matrix.
\end{lemma}

\begin{proof}
We begin by showing that all UUT-assemblages that come from a causal process matrix are causal. Let $\{w_{ab|xy}\}$ be such that
\beq
w_{ab|xy} = \tr_{A_IA_OB_IB_O}\left[(A^{A_IA_O}_{a|x}\otimes B^{B_IB_O}_{b|y}\otimes \1^{C_I})W^{A\prec B\prec C}\right], \ \ \ \forall \ a,b,x,y.
\eeq
Then, using \cref{eqWABC}, which is $_{C_I} W^{A\prec B\prec C}= _{B_OC_I} W^{A\prec B\prec C}$, it is possible to deduce that $\tr\sum_b w_{ab|xy}$ is independent of $y$, and hence $\{w_{ab|xy}\}=\{w^{A\prec B\prec C}_{ab|xy}\}$ is causally ordered. The equivalent is true for the order $B\prec A\prec C$. 

To prove the only if part we show that every causal UUT-assemblage can be reproduced by acting with some instruments on a causal process matrix. Given a causal UUT-assemblage, it can be decomposed into $\{w^{A\prec B\prec C}_{ab|xy}\}$ and $\{w^{B\prec A\prec C}_{ab|xy}\}$ with some convex weight $q$.

From $\{w^{A\prec B\prec C}_{ab|xy}\}$, we define\footnote{If $\tr(w^{A\prec B\prec C}_{ab|xy})=0$, we define $\rho_{ab|xy}$ as the null operator.}
\begin{align}
p^{A\prec B\prec C}(ab|xy) &\coloneqq \tr(w^{A\prec B\prec C}_{ab|xy}) \ \ \ \forall \ a,b,x,y, \\ 
p_A^{A\prec B\prec C}(a|x) &\coloneqq \sum_b \tr(w^{A\prec B\prec C}_{ab|xy}) \ \ \ \forall \ a,x, \\
\rho_{ab|xy} &\coloneqq \frac{w^{A\prec B\prec C}_{ab|xy}}{\tr(w^{A\prec B\prec C}_{ab|xy})} \ \ \ \forall \ a,b,x,y,
\end{align}
Using Bayes' rule we also have $p_B^{A\prec B\prec C}(b|axy)=p^{A\prec B\prec C}(ab|xy)/p_A^{A\prec B\prec C}(a|x)$. With this given quantities, we can construct instruments $\{A^{A\prec B\prec C}_{a|x}\}$ and $\{B^{A\prec B\prec C}_{b|y}\}$,
\begin{align}
A^{A\prec B\prec C}_{a|x} &= \1^{A_I}\otimes p_A^{A\prec B\prec C}(a|x)\ketbra{ax}{ax}^{A_O} \ \ \  \forall \ a,x \\
B^{A\prec B\prec C}_{b|y} &= \sum_{a,x}p_B^{A\prec B\prec C}(b|axy)\ketbra{ax}{ax}^{B_I}\otimes\rho_{ab|xy}^{T \ B_O} \ \ \ \forall \ b,y, 
\end{align}
and also the process matrix
\beq
W^{A\prec B\prec C} = \frac{\1^{A_I}}{d_{A_I}}\otimes\ketbra{\Phi^{+}}{\Phi^{+}}^{A_OB_I}\otimes\ketbra{\Phi^+}{\Phi^+}^{B_OC_I},
\eeq
and check that
\begin{align}
\tr_{A_IA_OB_IB_O}&\left[(A^{A\prec B\prec C}_{a|x}\otimes B^{A\prec B\prec C}_{b|y}\otimes \1^{C_I})W^{A\prec B\prec C}\right] =\\
&= p_A^{A\prec B\prec C}(a|x)p_B^{A\prec B\prec C}(b|axy)\rho_{ab|xy} \\
&=w^{A\prec B\prec C}_{ab|xy}
\end{align}
for every $a,b,x,y$. Analogously, the same holds for $\{w^{B\prec A\prec C}_{ab|xy}\}$.

Finally, just like in the proof of \cref{lmmbehaviour}, by allowing the different causal orders to act in complementary subspaces, we can recover any convex combinations of causally ordered assemblages, \textit{i.e.}, causal assemblages.
\end{proof}

\begin{theorem}
A tripartite process matrix is certified to be causally nonseparable in a semi-device-independent UUT way if and only if it can generate a noncausal UUT-assemblage for some choice of instruments for Alice and Bob.
\end{theorem}
The proof is analogous to \cref{thmsdi}.


\section{Proof that the quantum switch processes are causal in the UUT scenario}\label{app:switchUUTcausal}

In this appendix we prove \cref{thmswitch}, which we reestate for the convenience of the reader.
\setcounter{theorem}{6}

\begin{theorem}\label{thmswitch2}
The quantum switch processes cannot be certified to be causally nonseparable on a semi-device-independent scenario where Alice and Bob are untrusted and Charlie is trusted (UUT).

Moreover, any tripartite process matrix $W\in\set{L}(\set{H}^{A_IA_OB_IB_OC_I})$, with Charlie in the future of Alice and Bob, that satisfies
\begin{align}
\begin{split}
\tr[(A^{A_IA_O}_{a|x}\otimes &B^{B_IB_O}_{b|y}\otimes \1^{C_I})W^{A_IA_OB_IB_OC_I}] = \\
&q p^{A\prec B}(ab|xy) + (1-q) p^{B\prec A}(ab|xy),
\end{split}
\end{align}
for all $a,b,x,y$, where $0 \leq q \leq1$ is a real number, cannot be certified to be noncausal in a UUT scenario.
\end{theorem}

\begin{proof}
Let 
\beq
w^\text{switch}_{ab|xy} := \tr_{A_IA_OB_IB_O}\left[(A^{A_IA_O}_{a|x}\otimes B^{B_IB_O}_{b|y}\otimes \1^{C_I})W^\text{switch}\right]
\eeq
be a UUT-assemblage generated by the quantum switch. We define
\beq 
p(ab|xy) := \tr(w^\text{switch}_{ab|xy})
\eeq
and\footnote{If $\tr(w^\text{switch}_{ab|xy})=0$, we define $\rho_{ab|xy}$ as the null operator.}
\beq
\rho_{ab|xy} := \frac{w^\text{switch}_{ab|xy}}{p(ab|xy)}.
\eeq

Given that
\begin{align}
p(ab|xy) &= \tr(w^\text{switch}_{ab|xy}) \\
&= \tr\left[(A^{A_IA_O}_{a|x}\otimes B^{B_IB_O}_{b|y}\otimes \1^{C_I})W^\text{switch}\right] \\
&= q p^{A\prec B}(ab|xy) + (1-q) p^{B\prec A}(ab|xy),
\end{align}
since the quantum switch is device-independent causal, one can write
\begin{align}
w^\text{switch}_{ab|xy} &= p(ab|xy)\rho_{ab|xy} \\
&= q p^{A\prec B}(ab|xy)\rho_{ab|xy} + (1-q) p^{B\prec A}(ab|xy)\rho_{ab|xy}.
\end{align}
One can verify that the first term satisfies
\beq
\tr \Big(\sum_b p^{A\prec B}(ab|xy)\rho_{ab|xy}\Big) = \sum_b p^{A\prec B}(ab|xy) \tr(\rho_{ab|xy}) = p(a|x)
\eeq
and
\beq
\tr \Big(\sum_a p^{B\prec A}(ab|xy)\rho_{ab|xy}\Big) = \sum_a p^{B\prec A}(ab|xy) \tr(\rho_{ab|xy}) = p(b|y).
\eeq
Hence,
\beq
p^{A\prec B}(ab|xy)\rho_{ab|xy} = w^{A\prec B}_{ab|xy}
\eeq
is an assemblage causally ordered from Alice to Bob and
\beq
p^{B\prec A}(ab|xy)\rho_{ab|xy} = w^{B\prec A}_{ab|xy}
\eeq
is an assemblage causally ordered from Bob to Alice. Consequently,
\beq
w^\text{switch}_{ab|xy} = q w^{A\prec B}_{ab|xy} + (1-q) w^{B\prec A}_{ab|xy}
\eeq
is UUT causal for all instruments of Alice and Bob.

This proof holds for all process matrices W for which
\beq
p(ab|xy) = \tr\left[(A^{A_IA_O}_{a|x}\otimes B^{B_IB_O}_{b|y}\otimes \1^{C_I})W\right] \in \text{CAUSAL}.
\eeq
\end{proof}


\section{Analysis of experimental implementations of the quantum switch}\label{app:experimental}


In this appendix we present more details of our analysis of the experimental results involving the quantum switch reported in refs.~\cite{rubino17,goswami18-1}. Among other assumptions, both papers  consider a device-dependent scenario, where the analysis is made by assuming complete knowledge of all instruments involved (although no assumption is made about the process matrix). Here, we show that the experiment reported in both papers could have also certified indefinite causal order in a semi-device-independent scenario, where no assumptions are made about the instruments of one of the parties.	

We remark that (device-dependent) causal witnesses, including the ones of refs.~\cite{rubino17,goswami18-1}, are derived assuming that all instruments are implemented perfectly. However, due to experimental imperfections, this assumption is seldom true. One might then obtain a noncausal -- even nonprocess -- behaviour simply because the implemented instruments are not the ideal ones, rather than the data being genuinely noncausal (or nonprocess). Indeed, we have analysed the experimental data\footnote{We have analysed the data points of this experiment without taking the error bars into consideration.} collected in ref.~\cite{rubino17} and verified that if the instruments performed by Alice, Bob, and Charlie are trusted, the experimental behaviour is not a process behaviour. That is, there does not exist any process matrix $W$, causally separable or not, that can generate the experimental behaviour consistently with the assumed instruments\footnote{Moreover, even if we drop the assumption about the knowledge of the instrument performed by Charlie, there is no process matrix which is consistent with the experimental data collected.}. To properly anaylse experimental data we would need to allow some leeway in the instruments and probabilities, but developing the methods for that is beyond the scope of this paper.

Therefore, in order to avoid potential false positive results due to experimental error, we have considered the statistics one would have obtained in an ideal version of the experiment instead of considering the data collected on the actual experiments.
\\	

We start by analysing the device-dependent experiment described in ref.~\cite{goswami18-1} by Goswami \textit{et al.}, which considers an optical setup where three parties, Alice, Bob, and Charlie, have access to the reduced quantum switch process  $W_{\text{red}}$ (see \cref{eqswitch}). In this experiment,  Alice and Bob can choose between the 8 different unitary operations\footnote{Note that an unitary operation can be seen as an instrument with deterministic classical output.} and Charlie performs a measurement in the $\sigma_X$ basis on the control qubit state. Theoretically, the statistics of an ideal device-dependent TTT (trusted-trusted-trusted) experiment would be given by the behaviour
\beq
p^{\text{ideal1}}(c|x,y):=\tr\left[ \left(\overline{U}_x\otimes\overline{U}_y\otimes\overline{M}_c\right) W_{\text{red}}\right],
\eeq
where the instruments $\{\overline{U}_{x}\}$ and $\{\overline{M}_c\}$ are the ones described on the supplemental material of ref.~\cite{goswami18-1}. As pointed in ref.~\cite{goswami18-1}, the behaviour $\{p^{\text{ideal1}}(c|x,y)\}$ certifies indefinite causal order in a device-dependent way. Using the SDP formulation of the problem described in \cref{app:tripartite}, we show that the noisy behaviour
\beq
p_{\eta}^{\text{ideal1}}(c|x,y) := (1-\eta)\,p^{\text{ideal1}}(c|x,y) + \eta \, \frac{1}{2}
\eeq
where $p_I(c|x,y)=\frac{1}{2}$ is a uniform probability distribution, cannot be described by a causally separable tripartite process matrix, \textit{i.e.}, for some $c,x,y$,
\beq
p_{\eta}^{\text{ideal1}}(c|x,y) \neq \tr\left[ \left(\overline{U}_x\otimes\overline{U}_y\otimes\overline{M}_c\right) W^{\text{sep}}\right],
\eeq
in the range $\eta\in[0,0.1989)$. Hence, in this range of $\eta$, the behaviour $\{p_{\eta}^{\text{ideal1}}(c|x,y)\}$ certifies indefinite causal order in a device-dependent (TTT) way.

In order to analyse the behaviour $\{p^{\text{ideal1}}(c|x,y)\}$ in a semi-device-independent scenario, we drop the hypothesis that the measurement $\{M_c\}$ performed by Charlie is trusted, working in a TTU (trusted-trusted-untrusted) scenario. Using the SDP formulation of the problem described in \cref{app:tripartite}, we show that the noisy behaviour $\{p_{\eta}^{\text{ideal1}}(c|x,y)\}$ cannot be described by a causal TTU-assemblage, \textit{i.e.}, for some $c,x,y$,
\beq
p_{\eta}^{\text{ideal1}}(c|x,y) \neq \tr\left[ \left(\overline{U}_x\otimes\overline{U}_y\right)w_{c}^{\text{causal}}\right],
\eeq
in the same range of $\eta\in[0,0.1989)$. Hence, in this range of $\eta$, the behaviour $\{p_{\eta}^{\text{ideal1}}(c|x,y)\}$ certifies indefinite causal order in a semi-device-independent (TTU) way.

Therefore, the experimental setup described in ref.~\cite{goswami18-1} allows for certification of indefinite causal order with weaker hypotheses -- in a semi-device-independent scenario. Using the machinery developed in this work, we could not show that the behaviour $\{p^{\text{ideal1}}(c|x,y)\}$ can certify indefinite causal order in the UTT or TUU scenarios. However, since some of our SDP methods for the tripartite case may provide only an outer approximation of the sets of causal assemblages (see \cref{app:tripartite}), we cannot discard the possibility of such certification either. We remark, nevertheless, that we have proven that it is possible to certify that the switch process is causally nonseparable in these scenarios for the right choice of instruments (see \cref{sec:switch}).
\\

The experiment performed in ref.~\cite{rubino17} by Rubino \textit{et al.} considers a four-partite version of the switch process which is slightly more general than the tripartite one discussed in this paper. While we consider a tripartite switch process where the target state is embedded in the process, ref.~\cite{rubino17} considers a four-partite switch process in which the target state can be chosen by a fourth part in the common past of all other parties. If we label the fourth party $D_O$ (David output), the process matrix of the four-partite quantum switch is given by $W^4_{\text{switch}}:=\ketbra{w^4_{\text{switch}}}{w^4_{\text{switch}}}$,
where
\beq
\ket{w^4_{\text{switch}}}:=\frac{1}{\sqrt{2}}\left( \ket{\Phi^+}^{D_O A_I}  \ket{\Phi^+}^{A_O B_I}  \ket{\Phi^+}^{B_O C_I^t}\ket{1}^{C_I^c} 
+  \ket{\Phi^+}^{D_O B_I}  \ket{\Phi^+}^{B_O A_I}  \ket{\Phi^+}^{A_O C_I^t}\ket{0}^{C_I^c} \right)
\eeq
and $\ket{\Phi^+}:=\ket{00}+\ket{11}$ is an unnormalized maximally entangled two-qubit state. 

Reference \cite{rubino17} considers an optical setup where Alice, Bob, Charlie, and David have access to a reduced four-partite quantum switch process given by $W^4_\text{red}:=\tr_{C_I^t} \left(W^4_{\text{switch}}\right)$. In this experiment, David can choose between 3 different states $\overline{\rho}_d$, Alice can choose between 12 different dichotomic instruments $\{\overline{A}_{a|x}\}$, and Bob can choose between 10 unitary operations represented by the instruments $\{\overline{U}_y\}$. Theoretically, in an ideal device-dependent TTTT experiment, the statistics would be given by the behaviour
\beq
p^{\text{ideal2}}(a,c|x,y,d):=\tr\left[ \left(\overline{\rho}_d\otimes\overline{A}_{a|x}\otimes\overline{U}_y\otimes\overline{M}_c\right) W^4_\text{red}\right],
\eeq
where the instruments $\{\overline{\rho}_d\}$, $\{\overline{A}_{a|x}\}$, $\{\overline{U}_y\}$, and $\{\overline{M}_c\}$ are the ones described on the supplemental material of ref.~\cite{rubino17}.

In order to use the SDP formulations described in \cref{app:tripartite}, we restrict our analysis to the particular case where David outputs the state $\rho_1=\ketbra{0}{0}$. That is, we analyse the behaviour
\beq
p^{\text{ideal2}}(a,c|x,y,1):=\tr\left[\left(\overline{\ketbra{0}{0}}\otimes\overline{A}_{a|x}\otimes\overline{U}_y\otimes\overline{M}_c\right) W^4_{\text{red}}\right].
\eeq

Notice that when David is restricted to the choice of the state $\ketbra{0}{0}$, the four-partite reduced switch process $W^4_\text{red}$ relates to the tripartite reduced switch process $W_{\text{red}}$  (\cref{eqswitch}) via the identity
\beq
\tr_{D_O} \left[ \left(\ketbra{0}{0}^{D_O} \otimes \id^{A_I A_O} \otimes \id^{B_I B_O} \otimes \id^{C_I^c} \right) W^4_\text{red} \right] =  W_{\text{red}}.
\eeq	
Moreover, if the behaviour $\{p^{\text{ideal2}}(a,c|x,y,1)\}$ allows for certification of indefinite causal order, the full behaviour $\{p^{\text{ideal2}}(a,c|x,y,d)\}$ also allows for certification of indefinite causal order. We can, hence, use our tripartite machinery to certify indefinite causal order in this particular four-partite case.
	
Using the SDP formulation of the problem described in \cref{app:tripartite}, we show that the noisy behaviour
\begin{equation}
p_{\eta}^{\text{ideal2}}(a,c|x,y,1) := (1-\eta)\,p^{\text{ideal2}}(a,c|x,y,1) + \eta \, \frac{1}{4}
\end{equation}
where $p_I(a,c|x,y)=\frac{1}{4}$ is a uniform probability distribution, cannot be described by a causally separable tripartite process matrix, \textit{i.e.},  for some $a,c,x,y$,
\beq
p_{\eta}^{\text{ideal2}}(a,c|x,y,1) \neq \tr\left[\left(\overline{A}_{a|x}\otimes\overline{U}_y\otimes\overline{M}_c\right) W^{\text{sep}}\right],
\eeq
in the range $\eta\in[0,0.2300)$. Hence, in this range of $\eta$, the behaviour $\{p_{\eta}^{\text{ideal2}}(a,c|x,y,1)\}$ certifies indefinite causal order in a device-dependent (TTT) way.

We now analyse the behaviour $p^{\text{ideal2}}$ in a semi-device-independent scenario, where we drop the hypothesis that the measurement $\{M_c\}$ performed by Charlie is trusted, working in a TTU (trusted-trusted-untrusted) scenario. Using the SDP formulation of the problem described in \cref{app:tripartite}, we show that the noisy behaviour $\{p_{\eta}^{\text{ideal1}}(c|x,y)\}$ cannot be described by a causal TTU-assemblage, \textit{i.e.}, for some $a,c,x,y$,
\beq
p_{\eta}^{\text{ideal2}}(a,c|x,y,1) \neq \tr\left[\left(\overline{A}_{a|x}\otimes\overline{U}_y\right)w_{c}^{\text{causal}}\right],
\eeq
in the same range of $\eta\in[0,0.2300)$. Hence, in this range of $\eta$, the behaviour $\{p_{\eta}^{\text{ideal2}}(a,c|x,y,1)\}$ certifies indefinite causal order in a semi-device-independent (TTU) way.

Therefore, the experimental setup described in ref.~\cite{rubino17} allows for certification of indefinite causal order with weaker hypotheses -- in a semi-device-independent scenario.
Using the machinery developed in this work, we could not show that the behaviour $\{p^{\text{ideal2}}(a,c|x,y,d)\}$ can certify indefinite causal order in other semi-device-independent scenarios. However, since we have only considered the case where the state $\rho_z$ chosen by David is fixed and some of our SDP methods for the tripartite case may provide only an outer approximation of the sets of causal assemblages (see \cref{app:tripartite}), we cannot discard the possibility of such certification either. We remark, nevertheless, that we have proven that it is possible to certify that the switch process is causally nonseparable in other semi-device-independent scenarios for the right choice of instruments (see \cref{sec:switch}).

\end{document}